%% file: main.tex
\documentclass[acmsmall]{acmart}

\usepackage{enumitem}
\usepackage{xspace}
\usepackage{epsfig}
\usepackage{mathrsfs}
\usepackage{eucal}
\usepackage{multirow}
\usepackage{subfigure}
\usepackage{mathtools}

\DeclarePairedDelimiter\floor{\lfloor}{\rfloor}
\newtheorem{example}{Example}

\makeatletter  
\newif\if@restonecol  
\makeatother

\usepackage[linesnumbered,ruled,vlined]{algorithm2e}
\usepackage{algpseudocode}  
\usepackage{amsmath}  
\renewcommand{\algorithmiccomment}[1]{\bgroup\hfill//~#1\egroup}

\AtBeginDocument{
  \providecommand\BibTeX{{
    \normalfont B\kern-0.5em{\scshape i\kern-0.25em b}\kern-0.8em\TeX}}}

\setcopyright{acmcopyright}
\copyrightyear{2021}
\acmYear{2021}
\acmDOI{}

\acmJournal{JACM}
\acmVolume{1}
\acmNumber{1}
\acmArticle{1}
\acmMonth{1}

\begin{document}

\title{Let Trajectories Speak Out the Traffic Bottlenecks}

\author{Hui Luo}
\affiliation{
  \institution{RMIT University}
  \city{Melbourne}
  \country{Australia}}
\email{hui.luo@rmit.edu.au}

\author{Zhifeng Bao}
\affiliation{
	\institution{RMIT University}
	\city{Melbourne}
	\country{Australia}}
\email{zhifeng.bao@rmit.edu.au}

\author{Gao Cong}
\affiliation{
	\institution{Nanyang Technological University}
	\country{Singapore}}
\email{gaocong@ntu.edu.sg}

\author{J. Shane Culpepper}
\affiliation{
	\institution{RMIT University}
	\city{Melbourne}
	\country{Australia}}
\orcid{0000-0002-1902-9087}
\email{shane.culpepper@rmit.edu.au}

\author{Nguyen Lu Dang Khoa}
\affiliation{
	\institution{Data61, CSIRO}
	\city{Eveleigh}
	\country{Australia}}
\email{khoa.nguyen@data61.csiro.au}

\renewcommand{\shortauthors}{Luo et al.}

\ccsdesc[500]{Information systems~Geographic information systems}
\ccsdesc[500]{Information systems~Network data models}
\ccsdesc[500]{Theory of computation~Streaming models}

\keywords{Traffic Spread, Traffic Bottleneck, Road Segments Influence}

\input{macros.tex}
\input{sec0-abstract.tex}

\maketitle

\input{sec1-introduction.tex}

\input{sec2-related.tex}

\input{sec3-problem.tex}

\input{sec4-algorithm.tex}

\input{sec5-experiment.tex}

\input{sec6-conclusion.tex}

\input{sec7-ack.tex}


\bibliographystyle{ACM-Reference-Format}
\bibliography{ref_all}


\end{document}

%% file: macros.tex
\newcommand{\hui}[1]{{\color{red}{\bf{Hui says:}} \emph{#1}}}
\newcommand{\bao}[1]{{\color{blue}{\bf{Bao says:}} \emph{#1}}}
\newcommand{\myparagraph}[1]{\vspace{0.2\baselineskip}\noindent{\textbf{#1.}}~}
\newcommand{\tabincell}[2]{\begin{tabular}{@{}#1@{}}#2\end{tabular}}  


\newcommand{\problemlong}{traffic bottleneck identification\xspace}
\newcommand{\problemshort}{\ensuremath{\mathsf{TBI}}\xspace}
\newcommand{\mcproblemshort}{\ensuremath{\mathsf{MC}}\xspace}
\newcommand{\mkcproblemshort}{\ensuremath{\mathsf{MKC}}\xspace}

\newcommand{\trajectorySet}{\ensuremath{T}\xspace}
\newcommand{\trajectoryI}{\ensuremath{tr_i}\xspace}

\newcommand{\roadNetwork}{\ensuremath{R}\xspace}
\newcommand{\vertexSet}{\ensuremath{V}\xspace}
\newcommand{\vertexI}{\ensuremath{v_i}\xspace}
\newcommand{\vertexJ}{\ensuremath{v_j}\xspace}

\newcommand{\seedSet}{\ensuremath{S}\xspace}

\newcommand{\edgeSet}{\ensuremath{E}\xspace}
\newcommand{\edge}{\ensuremath{e}\xspace}
\newcommand{\edgeprime}{\ensuremath{e^\prime}\xspace}
\newcommand{\edgeIJ}{\ensuremath{e_{i,j}}\xspace}
\newcommand{\edgeJI}{\ensuremath{e_{j,i}}\xspace}

\newcommand{\edgeKI}{\ensuremath{e_{k,i}}\xspace}
\newcommand{\edgeJK}{\ensuremath{e_{j,k}}\xspace}

\newcommand{\weightSet}{\ensuremath{W}\xspace}
\newcommand{\weight}{\ensuremath{w}\xspace}
\newcommand{\weighIJ}{\ensuremath{w_{i,j}}\xspace}

\newcommand{\flow}{\ensuremath{f^0(\edge)}\xspace}
\newcommand{\flowIJ}{\ensuremath{f_\edgeIJ^t}\xspace}
\newcommand{\flowJI}{\ensuremath{f_\edgeJI^t}\xspace}

\newcommand{\diffusePro}{\ensuremath{p}\xspace}
\newcommand{\residualPro}{\ensuremath{r}\xspace}
\newcommand{\diffuseTime}{\ensuremath{\Delta t}\xspace}
\newcommand{\diffuseIter}{\ensuremath{\lambda}\xspace}
\newcommand{\diffuseDis}{\ensuremath{\delta}\xspace}
\newcommand{\spreadTime}{\ensuremath{w}\xspace}
\newcommand{\monitorTime}{\ensuremath{W}\xspace}
\newcommand{\duration}{\ensuremath{\tau}\xspace}

\newcommand{\resultNum}{\ensuremath{K}\xspace}
\newcommand{\resultNumPara}{\ensuremath{\epsilon}\xspace}

\newcommand{\timePara}{\ensuremath{\Delta T}\xspace}

\newcommand{\flowPara}{\ensuremath{\theta}\xspace}

\newcommand{\influScore}{\ensuremath{\Phi(\seedSet)}\xspace}

\newcommand{\seed}{\ensuremath{s}\xspace}
\newcommand{\seedSetOpt}{\ensuremath{S^*}\xspace}
\newcommand{\seedSetprime}{\ensuremath{\seedSet^\prime}\xspace}
\newcommand{\partitionlongname}{partition-based method\xspace}

\newcommand{\partitionNum}{\ensuremath{M}\xspace}
\newcommand{\weightmatrix}{\ensuremath{\eta}\xspace}

\newcommand{\phaseOne}{influence acquisition\xspace}
\newcommand{\phaseTwo}{bottleneck identification\xspace}
\newcommand{\ri}{\ensuremath{RI}\xspace}
\newcommand{\sampleSize}{\ensuremath{\beta}\xspace}

\newcommand{\chengdu}{\ensuremath{\textit{Chengdu}}\xspace}
\newcommand{\xian}{\ensuremath{\textit{Xi'an}}\xspace}
\newcommand{\porto}{\ensuremath{\textit{Porto}}\xspace}

\newcommand{\graphscaleThree}{0.31}
\newcommand{\graphscaleTwo}{0.47}

\newcommand{\topkmin}{\ensuremath{\mathsf{topKM}}\xspace}
\newcommand{\bfa}{\ensuremath{\mathsf{BF}}\xspace}
\newcommand{\sg}{\ensuremath{\mathsf{SG}}\xspace}
\newcommand{\pb}{\ensuremath{\mathsf{PB}}\xspace}
\newcommand{\cg}{\ensuremath{\mathsf{CG}}\xspace}
\newcommand{\cb}{\ensuremath{\mathsf{CB}}\xspace}

\newcommand{\imshortname}{\ensuremath{\mathsf{IM}}\xspace}
\newcommand{\imlongname}{influence maximization\xspace}

\newcommand{\icshortname}{IC\xspace}
\newcommand{\ltshortname}{LT\xspace}

%% file: sec0-abstract.tex
\begin{abstract}
{\em Traffic bottlenecks} are a set of road segments that have an
unacceptable level of traffic caused by a poor balance between road
capacity and traffic volume.
A huge volume of trajectory data which captures
realtime traffic conditions in road networks provides promising
new opportunities to identify the traffic bottlenecks.
In this paper, we define this problem as {\em trajectory-driven
traffic bottleneck identification}: Given a road network
{\roadNetwork}, a trajectory database {\trajectorySet}, find a
representative set of seed edges of size {\resultNum} of traffic
bottlenecks that influence the highest number of road segments not in
the seed set.
We show that this problem is NP-hard and propose a framework to find
the traffic bottlenecks as follows.
First, a traffic spread model is defined which represents changes in
traffic volume for each road segment over time.
Then, the traffic {\textit{diffusion probability}} between two
connected segments and the {\textit{residual ratio}} of traffic
volume for each segment can be computed using historical trajectory
data.
We then propose two different algorithmic approaches to solve the
problem. The first one is a best-first algorithm {\bfa}, with
an approximation ratio of $1-1/e$.
To further accelerate the identification process in larger datasets,
we also propose a sampling-based greedy algorithm {\sg}. 
Finally, comprehensive experiments using three different datasets
compare and contrast various solutions, and provide insights into
important efficiency and effectiveness trade-offs among the
respective methods.
\end{abstract}


%% file: sec1-introduction.tex
\section{Introduction}\label{sec-introduction}
Traffic congestion continues to be a pervasive problem in large
cities around the world, and is of particular concern in cities
experiencing the highest growth~{\cite{systematics2004traffic}}.
According to the Texas A\&M Transportation Institute's 2019 Urban
Mobility Report~{\cite{lasley2019}}, traffic congestion 
and unexpected delays cost \$$18.12$ per person per hour.
{\em Traffic bottlenecks} are one of the leading
causes of congestion~{\cite{hale2016traffic}}.

A traffic bottleneck is a localized congestion problem which is
often caused by a small number of road segments which converge
somewhere in the road network.
When the traffic volume on bottleneck edges overwhelm the surrounding
road capacities, then congestion increases within the road
network and reduces the traffic flow throughout a transportation
network~{\cite{ji2014empirical, zheng2014urban}}.
The detection and removal of traffic bottlenecks therefore are a
management priority in road networks, and can be used to help a
traffic management company to reduce traffic congestion through new
construction or traffic signal
timing~{\cite{yuan2014identification}}.
However, dynamically identifying traffic bottlenecks remains to be an
important and unsolved problem since many different factors influence
traffic movement in the network, which can change suddenly or slowly
over time~{\cite{yue2018urban}}.

Bottlenecks in high demand areas can result in traffic
{\textit{diffusion}} described by a traffic spread model, which will
influence neighboring road segments and will, in turn, create new
traffic bottlenecks as traffic continues to move across a road
network~{\cite{bertini2006you}}.
Intuitively, when any road segment is heavily congested, traffic
migrates to neighboring roads, causing the neighboring roads to
become bottlenecks.
Two other important issues are common when attempting to create a
data-driven traffic spread model~{\cite{anwar2020influence}}: 
(1) {\textit{faulty sensors}}: Realtime sensor data is often noisy,
which can make the predictions less reliable.
(2) {\textit{missing records}}: some segments in a road network may
have no sensor data available at all.
For example, $10$\% of daily traffic volume data in Beijing is not
reliable or missing entirely~{\cite{qu2009ppca}}.
Relying purely on a node-edge road network graph is not sufficient,
since the shortest path is not always the route taken by
a commuter~{\cite{xu2018discovery}}.

Nevertheless, a huge volume of trajectory data is being generated by
the transportation industry at an unprecedented
rate~\cite{zheng2015trajectory, wang2021survey, wang2018torch}, which captures realtime traffic
conditions in road networks and provides promising new opportunities
that can be used to address this longstanding problem.
In this work, we propose a promising new approach to resolve this
important challenge, which we refer to as the Trajectory-driven
{\underline{T}raffic} {\underline{B}ottleneck}
{\underline{I}dentification} problem (\problemshort): Given a
directed road network represented with a graph, and a database of
trajectories, the {\problemlong} problem finds {\resultNum} edges in
a graph, called {\textit{seeds}}, such that the total number of edges
influenced is maximized for the traffic spread model.
Less formally, when given an existing road network with vehicles at a
certain time, improving the traffic conditions of the {\resultNum}
road segments selected would result in the largest improvement in
traffic across the entire road network.
Therefore, the key question becomes: which set of edges should be
targeted?
While we do provide a solution for the problem as defined above, we
will not explore how professionals might best use this information to
improve traffic conditions given these road segments in this work.

In order to solve the {\problemshort} problem, we must first overcome
several important challenges.
Firstly, {\textit{how should traffic be modeled dynamically over time}}?
Traffic congestion can cascade through the edges of the network in
unexpected ways~{\cite{gomez2012inferring}}.
A traffic spread model can be used to represent how traffic flow
diffuses through the network over time.
Existing methods assign propagation rates between
two adjacent edges include, and can be: a {\textit{constant}} value
(e.g., $0.1$), drawn {\textit{uniformly}} from a predefined set of
values (e.g., $\{0.1, 0.01, 0.001\}$), or represented by the
reciprocal of the node {\textit{degree}}~{\cite{goyal2011data}}.
However, these methods do not capture patterns in real data.
Secondly, {\textit{how do you identify {\resultNum} seed edges (the
traffic bottlenecks) based on the traffic spread model}}?
Each bottleneck edge candidate can influence a set of neighboring
edges which also become congested, and the ``influence'' from several
bottleneck edges can overlap.
Let $I(\{e\})$ denote the influence of the edge $e$.
For instance, $I(\{e_1\})=\{e_3, e_4\}$, $I(\{e_2\})=\{e_4,
e_5\}$, both $e_1$ and $e_2$ will influence $e_4$.
However, when we consider $e_1$ and $e_2$ into the seed edges set,
$e_4$ will only be counted once.
Thus, the {\resultNum} seed edges cannot be obtained by simply
ranking candidates by influence to produce the best candidate set.

To resolve the first challenge, we propose that each road segment can
be modeled as a bidirectional weighted edge, where the weight is
equal to the traffic volume, and is closely related to a temporal
factor~{\cite{huang2014deep,lv2014traffic}}.
For example, more vehicles traveling from the suburb to the CBD
happen in the morning peak hours, whereas more vehicles travel in the
opposite direction in the evening.
In addition, we maintain two important characteristics when we
construct our traffic spread model, namely the traffic
{\textit{diffusion} probability} between two adjacent edges and the
{\textit{residual} ratio} of traffic volume for each edge, which can be
computed using historical trajectory data.
The traffic spread model has two important properties: (1)
{\textit{spatial influence range}}, which constrains how far the
diffusion chain can propagate to other edges;
(2) {\textit{scale}}, which is the number of affected edges.
Specifically, we will use the scale value to represent the influence
of each edge.

To resolve the second challenge, we maintain a {\textit{monitor time
window}} {\monitorTime} and transform a historical timeline into
snapshots, where the size of each snapshot is the {\textit{spread
time window}} {\spreadTime}.
As time progresses, the result is updated for {\monitorTime} and the
traffic volume is updated every {\spreadTime} using the specified
traffic spread model.
Next, a two phase algorithm, {\phaseOne} and {\phaseTwo}, first
identifies the influenced edges, and then selects {\resultNum} seed
edges with the maximum coverage for the traffic bottleneck edges
with the most influence over the road network.
Specifically, when we select a seed edge, we apply a best-first
(\bfa) algorithm.
Our main idea is to iteratively find the most profitable edge among
all unselected edges.
Though the idea is simple, the effectiveness is consistently better.
However, when the number of edges is large, it becomes computationally
expensive to search the most profitable edge.
To improve the efficiency, we also propose a sampling-based greedy
(\sg) algorithm, which performs the selection on a small subset of
sampled candidate edges.

In summary, we list our contributions as below:
\begin{itemize}
\item We formalize the Traffic Bottleneck Identification
(\problemshort) problem and show that the problem is NP-hard in
Section~{\ref{sec-problem}.}
\item We propose a two-phase approximation algorithm which can
be used to solve the problem in  in Section~{\ref{sec-approach}}.
\item We perform the experimental study to investigate the
performance of our proposed algorithms, and compare them with the
state-of-art methods on three different test collections in
Section~{\ref{sec-exp}}.
\end{itemize}

Additionally, we discuss related work in Section \ref{sec-related}
and conclude the paper in Section \ref{sec-conclusion}.

%% file: sec2-related.tex
\section{Related Work}\label{sec-related}
In this section, we describe literature related to the traffic
spread model for traffic diffusion.
Then we introduce the \imlongname problem which is the most relevant
domain with ours in Section \ref{sec-im}.
\subsection{Traffic Spread}\label{sec-spread}
A traffic spread model captures traffic flow evolution in a road
network temporally \cite{tedjopurnomo2020survey}.
{\citet{anwar2020influence}} devise a probabilistic traffic diffusion
model and show how to compute the influence scores for each road
segment in an urban road network in order to select the top-$K$ edges
with the maximal influence scores.
\citet{long2008urban} propose a congestion propagation
model based on a cell transmission model, which describes flow
propagation using links, where a link represents a road, which is
divided into road segments called cells.
The congestion propagation model is validated using a simulated
test environment.
However, the model differs from our own in the following two
ways: (1) Their model only measures the \textit{aggregated
inflow} (total traffic volume) is propagated
into a cell, and it is unclear
if the \textit{road-to-road inflow} (which cell is the 
propagation originator) at the road segment level.
(2) Their model returns a segment of a road or a zone in the network
by comparing the speed to a configurable threshold, while we model the
relationships between all road segments to identify traffic
bottlenecks.
\citet{zhao2017data} study a traffic congestion diffusion model
for both both temporal and spatial data.
Each spatial region is divided into grids, and then the vehicle
traffic flow in and out of each one is computed dynamically in order
to construct a model of traffic flow between grids.
However, their model only computes the traffic spread between any two
grids and not two edges as in our work.
To capture this dynamic and temporal property,
{\citet{rodriguez2011uncovering}} propose the use of
probability-based diffusion models such as {\textit{Exponential}}
{\textit{Power law}}, and {\textit{Rayleigh} modeling.
{\citet{saberi2020simple}} devise an epidemic framework inspired by
the susceptible-infected-recovered (SIR) model to describe the
dynamics of traffic congestion spread, and support their claim based
on a historical multi-city analysis.
Each node has three states: {\textit{susceptible}},
{\textit{infected}}, or {\textit{recovered}} to simulate whether a
road state is ``contaminated'' or ``recovered''.
If a node is {\textit{susceptible}}, then it has never been
contaminated.
If a node is {\textit{infected}}, then it is currently contaminated.
However, if a node is {\textit{recovered}}, then it was contaminated
and has now recovered.
Although the model captures the congested road congestion over time,
it has two disadvantages: (1) it does not show the specific roads
which are in a susceptible, infected, or recovered state;
%
%
(2) modeling influence between two road segments remains an open
problem when using this model.

To summarize, our \problemshort problem differs in the following
ways: (1) We aim to define a data-driven traffic spread model for a
large city by profiling and analyzing the trajectory datasets, and
do not rely on a mathematical simulation model, such as done by
\citet{rodriguez2011uncovering}.
(2) When considering the influence of each road segment, we are not
limited to the traffic measures shown in the works
\cite{long2008urban, anwar2020influence, saberi2020simple}, as road
sensors may be faulty or even missing.
Moreover, evaluating the effects of road segments cannot be simply
reduced by comparing the traffic volume or speed to a predefined
threshold~\cite{bertini2006you, rao2012measuring}.
Instead, we focus on modeling the influence between any two road
segments, which is rarely considered in previous work.



\subsection{Influence Maximization and Variations}\label{sec-im}
In the problem of {\imlongname} ({\imshortname}), the goal is to find a seed
set composed of \resultNum nodes that maximize influence spread over a
network graph for a given influence model~{\cite{tang2014influence}}.
The two most commonly used influence models
are~{\cite{kempe2003maximizing}}: (1) Independent Cascade (IC): each
node may be active or inactive initially and time proceeds at
discrete timestamps.
At step $t$, every node $v$ that became active at step $t-1$
activates a non-active neighbor $w$ with probability $p$.
If it fails, it does not try again; 
(2) Linear Threshold (LT): Similar to IC model, each node may be
active or inactive at the beginning.
In addition, each directed edge has a weight $w$, and each node has a
randomly generated threshold value $\tau$.
At step $t$, an inactive node becomes active if the aggregated
weights of all the neighborhood nodes within one hop can exceed the
threshold $\tau$.
For further details see the recent survey of
{\citet{li2018influence}}.

The IC model is not applicable to the traffic diffusion problem
since, under assumptions imposed by the IC model, each edge can only
be considered once, and activated with a fixed probability, which is
not true for traffic flow.
Besides, the {\icshortname} model does not rely on edge
weights~{\cite{wang2010community}}.
Although the {\ltshortname} model considers edge weights, there is
currently no approach to assign the probabilities using a
dataset~{\cite{li2018influence}}.
In summary, our {\problemshort} problem differs from the {\imshortname}
problem in several fundamental ways: (1) In {\imshortname}, the state
of a node can be switched from being {\textit{inactive}} to being
{\textit{active}}, but not vice versa~{\cite{jiang2011simulated}}.
However, the original \textit{active} (i.e., a bottleneck edge) edge
may become \textit{inactive} (i.e., a non-bottleneck edge) as time
elapses.
(2) The propagation of influence from one road to another road may
incur a certain {\textit{time delay}}~{\cite{chen2012time}}, while
the temporal information is rarely considered when using the
{\icshortname} and {\ltshortname} models.

\subsection{Other related areas}\label{sec-other}
In addition to the literature above, there are some other domains which are slightly related to our work, such as spatial object selection \cite{guo2016influence, zhang2018trajectory, zhang2019optimizing, guo2018efficient} and maximizing bichromatic reverse k nearest neighbor (MaxR$k$NN) problem \cite{choudhury2016maximizing, luo2018maxbrknn, zhou2011maxfirst}.
Specifically, \citet{guo2016influence} define an \imlongname problem on
trajectories, which aims to find a subset of trajectories with the
maximum expected influence among a group of audiences.
\citet{zhang2018trajectory} propose a billboard placement problem 
and find a set of billboards when given a budget which influence the
most trajectories.
\citet{zhang2019optimizing} extend the billboard problem of 
\citet{zhang2018trajectory} to support the inclusion of impression
counts.
\citet{guo2018efficient} study how to select a set of representative spatial objects from the current region of users' interest.
The MaxR$k$NN query aims to find the locations to set up new facilities which can serve the most number of users assuming that users prefer to go to the nearest facility. However, neither of these works use an influence model which can
be applied in our problem scenario.

%% file: sec3-problem.tex
\section{Problem Formulation}\label{sec-problem}
\input{tbl-notation.tex}
In this section, we formalize the {\problemlong} problem.
Table~{\ref{tbl-notation}} summarizes the necessary notations.

\myparagraph{Road Network} A road network $\roadNetwork=\langle
\vertexSet, \edgeSet \rangle$ can be represented as a directed and
edge-weighted graph, which contains a set of vertices \vertexSet and
a set of edges \edgeSet.
A directed edge \edgeIJ $(\edgeIJ \in \edgeSet)$ is a road segment
from a vertex \vertexI ($\vertexI \in \vertexSet$) to another
connected vertex \vertexJ ($\vertexJ \in \vertexSet$).
Similarly, the directed edge from the vertex \vertexJ to \vertexI can
be denoted as \edgeJI.
For ease of presentation, we may omit \edgeIJ subscripts and
use $e$ to denote a directed edge when the context is clear.
Each edge \edge has a weight $f(\edge)$ to represent the traffic
volume of a certain road segment \edge.

\myparagraph{Trajectory} Let \trajectorySet denote a set of
trajectories $\langle tr_1, tr_2, ...  \rangle$.
A trajectory is a sequence of timestamped geo-coordinates:
$\trajectoryI=\{(o_i^1, t_i^1), (o_i^2, t_i^2), ..., (o_i^n,
t_i^n)\}$, where each $o_i^j$ ($o_i^j \in \vertexSet, 1 \leq j \leq
n$) is a spatial location in a $2$-dimensional data space, and each
$t_i^j$ ($t_i^1 \textless t_i^2 \textless ...
\textless t_i^n$) is the timestamp at which the corresponding $o_i^j$
occurs.

Assume that we have a traffic monitor time window $W$ ($1$ hour for
example) and a traffic spread time window $w$ (such as $20$ seconds),
the spread time window will slide every $w$ time until the deadline
of the monitor time window.
The spread time window moves $\eta = \floor*{\frac{W}{w}}$ 
batches, and then a set of seed edges can be computed for this
time period.
The traffic volume of each edge is updated for time {\spreadTime}.


\subsection{Traffic Spread Model}
Before we define our problem formally, we first explain the traffic
spread model to indicate how traffic volume diffuses through a 
road network.
Our main goal is to determine how the traffic spread influences
changes in traffic volume for each edge (as described in
Definition~{\ref{def-diffusion}}).
First, we define how an edge can influence another edge.

\begin{definition}
\label{def-influ}
\myparagraph{Influenced Value} Given an edge \edge with an initial
traffic volume value $f^0(\edge)$, the influence value $Inf(\edge,
\edgeprime)$ of \edgeprime by \edge after one spread time window $w$
is defined as:
\begin{equation}
\label{equ-influ}
Inf(e, e^\prime) = \left\{\begin{matrix}
p(e, e^\prime)f^0(e)(1-r(e))r(e^\prime) & \text{If }t(e, e^\prime) \leq w	\\ 
0 & \text{Otherwise}
\end{matrix}\right.
\end{equation}
Here, $t(e, e^\prime)$ denotes the travel time from $e$ to $e^\prime$
along a road network.
Now consider that traffic spread models have a limited coverage region 
for their influence~{\cite{li2018diffusion, yu2018spatio}}.
So, we assume that if $t(e, e^\prime) \leq w$, then the influence
from $e$ to $e^\prime$ is valid, and
$\diffusePro(\edge, \edgeprime)$ is the traffic diffusion probability
from \edge to \edgeprime.
If two edges \edge and \edgeprime are connective neighbors, then
$\diffusePro(\edge, \edgeprime)$ can be directly derived from the
real historical trajectory statistics and computed using the traffic
volume from \edge to \edgeprime divided by the total traffic volume
diffused out from \edge.
Otherwise, $\diffusePro(\edge, \edgeprime)$ can be computed as
$\sum_{}^{}(\diffusePro(\edge, \edge_1)\dots \diffusePro(\edge_n,
\edgeprime)\prod_{i=1}^{n}(1-\residualPro(e_i))$), where $\langle
\edge, \edge_1, \dots, \edge_n, \edgeprime \rangle$ is a travel path
from \edge to \edgeprime along the road network.
Here, $\sum(\cdot)$ denotes the aggregated probabilities of all
possible paths within $w$ travel time from \edge to \edgeprime,
and $\residualPro(e)$ is the residual ratio of traffic
volume to denote the remaining traffic volume which may still reside
in edge $e$.
\end{definition}

\input{tbl-road.tex}
\vspace{3ex}
\begin{example}
\label{exa-road}
Figure \ref{fig-road-example} shows an example of a road network,
which consists of eight nodes from $v_1$ to $v_8$, and ten directed
edges.
For ease of visualization, we assume that only one directed edge
exists for any two connected nodes, for example the edge from $v_1$
to $v_2$ is denoted as $e_{1,2}$.
We use dotted lines to represent exiting edges, such as
the dotted edge originating from $v_6$.
In addition, Table~{\ref{tbl-pr}} describes the traffic diffusion
probability $\diffusePro(\edge, \edgeprime)$ between any two directly
connected edges \edge and \edgeprime in a certain period of time.
For instance, $\diffusePro(e_{2,3}, e_{3,4})=0.3$.
We assume the residual ratio $\residualPro(\edge)$ equals to $0.5$ in
the following examples.
Interestingly, both the $\diffusePro(\edge, \edgeprime)$ and
$\residualPro(\edge)$ values can be estimated by analyzing the
historical trajectory datasets and depend on the time window we aim
to monitor.
In our experiments, we precompute $\diffusePro(\edge, \edgeprime)$
and $\residualPro(\edge)$ values for each hour.
For example, using the historical trajectory dataset we know that
there are $50$ vehicles on the edge \edge initially.
After traffic diffusion occurs, $20$ vehicles are spread out while $30$ vehicles are still left on the edge \edge.
We also know that $10$ vehicles moved from \edge to \edgeprime.
So, we can compute: $\residualPro(\edge)=\frac{30}{50}$, and
$\diffusePro(\edge, \edgeprime) = \frac{10}{30}$.
In practice, the roads are connected and the traffic diffusion
probability \diffusePro is normally less than one.
\end{example}
\begin{definition}
\label{def-diffusion}
\myparagraph{Traffic Spread}
We now describe the traffic volume change after one traffic spread
window to ease illustration.
Given an edge \edge with an initial traffic volume value
$f^0(\edge)$, the incremental traffic inflow $\Delta f^{in}$ diffused
by other edges is defined as:
\begin{equation}
\Delta f^{in}(\edge) = \sum_{\edgeprime \in \edgeSet} Inf (\edgeprime, \edge)
\end{equation}

Similarly, the incremental traffic outflow $\Delta f^{out}$ which
indicates the diffusing value from \edge to other influenced edges
is:
\begin{equation}
\Delta f^{out}(\edge) = f^0(\edge)(1-\residualPro(\edge) )= \sum_{\edgeprime \in \edgeSet} Inf (\edge, \edgeprime)
\end{equation}

Finally, if \edge has an initial travel volume $f^0(\edge)$, then the
traffic volume $f(\edge)$ after one traffic spread model is:
\begin{equation}
\label{equ-flow}
f(\edge) = f^0(\edge)\residualPro(\edge) + \Delta f^{in}(\edge)
\end{equation}
\end{definition}

\begin{example} Table~{\ref{tbl-diffuse}} illustrates traffic volume
after traffic diffusion occurs.
Note that, we assume an edge can influence its neighbor edges within
a single hop in this example to ease illustration, but the traffic
spread time window \spreadTime (say $20$ seconds) is used to restrict
the spread range in our experiments.
For example, the initial traffic volume of the edge $e_{2,3}$ is $40$.
After the diffusion, the incremental incoming traffic volume diffused
by the neighbor edges $e_{1,2}$ and $e_{4,2}$ is: $\Delta
f^{in}(e_{2,3})=1\times(1-0.5)\times 20+1\times(1-0.5)\times16=18$.
Similarly, the incremental out-degree traffic volume impacting the
neighbor edges $e_{3,4}$ and $e_{3,6}$ are: $\Delta
f^{out}(e_{2,3})=0.5\times40=20$.
Finally, the traffic volume of edge $e_{2,3}$ after one spread time
window is: $f(e_{2,3})=20+18=38$.
\end{example}
\input{tbl-diffusion.tex}
\begin{definition}
\label{def-bottleneck}
\myparagraph{Congested Edge}
If the traffic volume of an edge \edge is greater than a traffic
volume threshold, i.e., $f(\edge) \geq \flowPara \times len(e)$, then
\edge is considered to be a congested edge.
The length of the road segment, $len(e)$, can be easily computed
based on real world geographical coordinates for datasets collected
for a city.
The traffic volume threshold can be changed using the traffic
congestion parameter {\flowPara}.
\end{definition}

\begin{definition}
\label{def-influEdge}
\myparagraph{Influenced Edge}
Given two congested edges \edge and \edgeprime, if $Inf(e, e^\prime)
> 0$ and $f(e^\prime) \geq \flowPara \times len(e^\prime)$ (i.e., \edgeprime is congested), then we
call \edgeprime as an influenced edge of \edge.
Let $I_j(\{e\})$ denote the influenced edges of \edge at the $j$-th
spread time window, then $\edgeprime \in I_j(\{\edge\})$.
\end{definition}

\begin{definition}
\label{def-bottleneck-new}
\myparagraph{$\tau$-consecutive}
Given a set of influenced edges $I_1(\{\edge\})$, $I_2(\{\edge\})$,
$\ldots$  of the edge {\edge} for multiple spread time windows,
an influenced edge {\edgeprime} is $\tau$-consecutive
if it can be influenced by \edge that is consecutively no smaller
than the time window $\tau$ (the unit is a spread time window size).
The influenced edge set $I(\{\edge\})$ contains the influenced edges
which satisfy $\tau$-consecutive constraint.
\end{definition}

Note that, in contrast to the concept of an influenced edge as in
Definition ~{\ref{def-influEdge}} from only one spread time
window, the $\tau$-consecutive in
Definition~{\ref{def-bottleneck-new}} is from the perspective of the
monitor time window {\monitorTime} which covers multiple spread time
windows {\spreadTime}.

\begin{example} Suppose we have the influenced edges set
$I_1(\{\edge\})=\{e_1, e_2, e_3\}$, $I_2(\{\edge\})= \{e_3, e_4\}$,
$I_3(\{\edge\})=\{e_2, e_5\}$ for three spread time windows, and
$\tau = 2 \spreadTime$ (two spread time windows), then $I(\{\edge\})
= \{e_3\}$ since $e_3$ appears in $I_1(\{\edge\})$ and
$I_2(\{\edge\})$.
Although $e_2$ is influenced by the first and third spread time
windows, these two spread time windows are not consecutive.
\end{example}

\subsection{Traffic Bottleneck Identification}
Now, using the road network information and the traffic spread model
above, we are in a position to define the {\problemlong} problem.
The overall goal is to find the important edges in a road
network using the movement of traffic.


\begin{definition}
\label{def-prob}
\textbf{\underline{T}raffic \underline{B}ottleneck
\underline{I}dentification (\problemshort)} Given a road network
$\roadNetwork=\langle \vertexSet, \edgeSet \rangle$, a set of
trajectories \trajectorySet, a positive parameter \resultNum, a
traffic congestion parameter \flowPara, a consecutive time interval
$\tau$, a spread time window $w$ and a monitor time window $W$,
choose a set \seedSet of seed edges of size $\resultNum$ ($\resultNum
\ll |\edgeSet|$), where $\seedSet \subseteq \edgeSet$, such that the
total number of influenced edges \influScore is maximal, where
$\Phi(S)$ is:
\begin{equation}
\label{equ-obj}
\Phi(S)= |I(S)| = \sum_{e \in S}^{}|\cup I(\{e\})|
\end{equation}
\end{definition}

Here, $\Phi(S)$ is defined as the total number of edges influenced by
$S$.
However, additional factors such as the influence contribution degree
for each edge may also be included, and do not change the time
complexity of the problem.


\subsection{NP-Hardness}
We now show that the \problemshort problem is NP-hard using a
reduction from the Maximum $K$-Coverage problem.

\begin{lemma}
The \problemshort problem is NP-hard.
\end{lemma}
\begin{proof}
In the Maximum $K$-Coverage (\mkcproblemshort) problem, given a
collection of sets $C=\{C_1, C_2, ..., C_m\}$ over a set of objects
$O$, where $C_i \subseteq O$ and each element $o \in O$ has an
associated weight $w(o)$, and a positive integer $K$, we wish to know
whether {\resultNum} sets exist (e.g., $C_1, C_2, ..., C_K$) when 
the weight of the elements in $\cup_{i=1}^{K} C_i$ is maximized.
The proof for the NP-Hardness of the {\mkcproblemshort problem} 
by {\cite{hochbaum1998analysis}} can now be used as the target
reduction, and can be reduced as follows.
Given an arbitrary instance of the {\problemshort} problem, each
edge $e_i \in \edgeSet$ in the {\problemshort} problem can be mapped
the elements $o \in O$ in the {\mkcproblemshort} problem.
Each edge $e_i$ has an influenced edge set $I(\{e_i\})$, then we map
$I(\{e_i\})$ and the cardinality of $I(\{e_i\})$ ($|I(\{e_i\})|$ ) in
the \problemshort problem to the subset $C_i$ and the weight $w(o)$
in the \mkcproblemshort problem, respectively.
The goal of \problemshort problem as defined in
Equation~{\ref{equ-obj}} is to select $K$ edges from \edgeSet, such
that the total number of covered edges $\cup_{i=1}^{K} I(e_i)$ by the
selected seed edges set \seedSet ($|\seedSet| = K$) is maximized. 

So, solving the \problemshort problem is equivalent to deciding
whether there are $K$ sets whose union has the maximum aggregated
weight in the \mkcproblemshort problem.
Therefore, the \problemshort problem is NP-hard.
\end{proof}

%% file: tbl-notation.tex
\begin{table}[t]
	\small
	\renewcommand{\arraystretch}{1.2}
	\caption{Symbol and description}
	\label{tbl-notation}
	\vspace{1ex}
	\centering
	\begin{tabular}{lp{10cm}}
		\hline
		{\bfseries Symbol} & {\bf Description}	\\
		\hline
		
		$\roadNetwork = \langle \vertexSet, \edgeSet \rangle$	&	A road network with a set \vertexSet of vertices, and a set \edgeSet of edges	\\ 
		
		\trajectorySet	&	A set of trajectories, each of which contains a sequence of timestamped geo-coordinates	\\
	
		\edge, \edgeIJ, \edgeJI	&	A directed edge in \edgeSet	\\
		
		$f(\edge)$	&	The traffic volume of an edge \edge	\\
		
		\seedSet	&	A set of seed edges, $\seedSet \subseteq \edgeSet$, and $|\seedSet|=\resultNum$	\\
		
		$\diffusePro(e, e^\prime)$	&	The traffic diffusion probability from the edge \edge to $e^\prime$	\\
		
		$\residualPro(e)$	&	The residual ratio of traffic volume residing in the edge \edge \\
		
		$I(\{\edge\})$, or $I(\seedSet)$	&	The influenced edges impacted by an edge \edge, or a seed edge set \seedSet	\\
		
		$\Phi(\{e\})$, or \influScore	&	The influence score of a seed edge \edge, or a seed set \seedSet	\\
		
		\flowPara	&	The traffic congestion parameter \\
		
		$\tau$	&	The consecutive influence time threshold to measure a traffic bottleneck \\ 
		
		\spreadTime	&	The traffic spread time window	\\
		
		\monitorTime&	The traffic monitor time window	\\
		
		\hline
	\end{tabular}
\end{table}

%% file: tbl-road.tex
\makeatletter\def\@captype{figure}\makeatother
\noindent
\begin{minipage}{0.49\textwidth}
	\caption{Road network}
	\vspace{-1ex}
	\centering
	\label{fig-road-example}
	\includegraphics[scale=0.60]{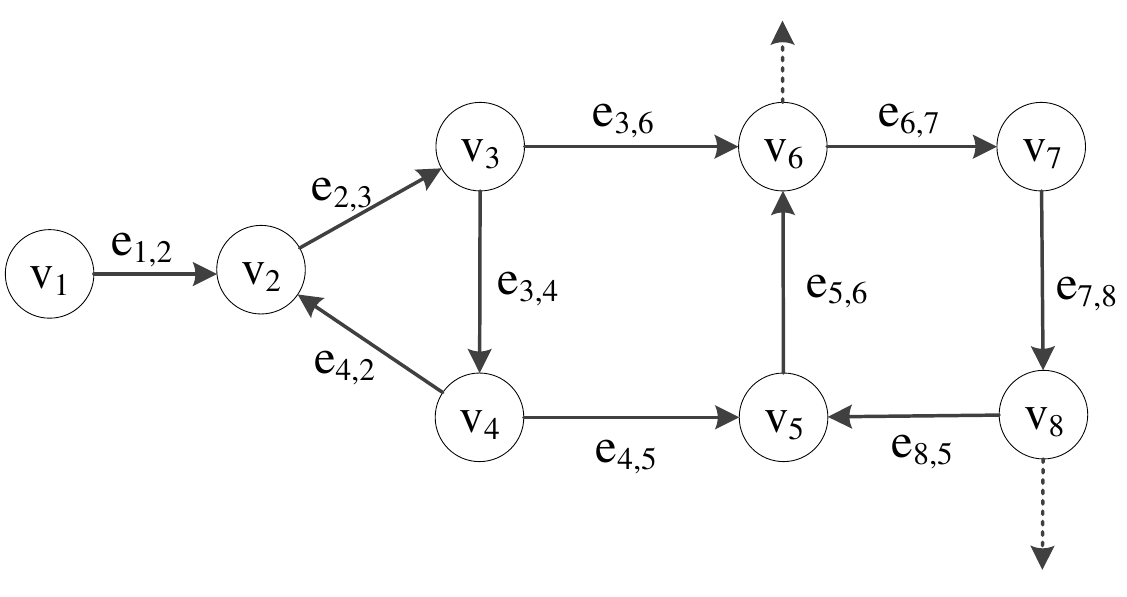}
\end{minipage}
\makeatletter\def\@captype{table}\makeatother
\begin{minipage}{0.49\textwidth}
	\small 
	\renewcommand{\arraystretch}{1.2}
	\caption{Diffusion probability \diffusePro(\edge, \edgeprime)} 
	\vspace{-2ex} 
	\centering
	\label{tbl-pr}
	\begin{tabular}{ccc|ccc}
		\hline
		\edge	&	\edgeprime	&	\diffusePro(\edge, \edgeprime)	&	\edge	&	\edgeprime	&	\diffusePro(\edge, \edgeprime)	\\	\hline
		$e_{1,2}$	&	$e_{2,3}$	&	1     &	  	  $e_{4,2}$	  &	$e_{2,3}$	&	1	 \\
		$e_{2,3}$	&	$e_{3,4}$	&	0.3	&		$e_{4,5}$	&	$e_{5,6}$	&	1	\\
		$e_{2,3}$	&	$e_{3,6}$	&	0.7	&		$e_{5,6}$	&	$e_{6,7}$	&	0.5	\\
		$e_{3,4}$	&	$e_{4,2}$	&	0.4	&		$e_{6,7}$	&	$e_{7,8}$	&	1	\\
		$e_{3,4}$	&	$e_{4,5}$	&	0.6	&		$e_{7,8}$	&	$e_{8,5}$	&	0.5 \\
		 $e_{3,6}$	  &	  $e_{6,7}$	  &	 0.5 &	  	  $e_{8,5}$	  &	$e_{5,6}$	&	1 \\
		\hline
	\end{tabular}
\end{minipage}

%% file: tbl-diffusion.tex
\begin{table}[h]
	\small
	\renewcommand{\arraystretch}{1.2}
	\caption{An example of traffic diffusion for one spread time window} 
	\label{tbl-diffuse}
	\vspace{-3ex} 
	\centering  
	\subtable[The initial traffic volume $f^0(e)$.]{ 
		\begin{tabular}{cc|cc}  \hline
			\edge	&	$f^0(e)$	&	\edge	&	$f^0(e)$   \\	\hline
			$e_{1,2}$	&	20	&	$e_{4,5}$	  &	18	\\
			$e_{2,3}$	&	40	  &	  $e_{5,6}$	  &	 36	\\	
			$e_{3,4}$	&	10	&	$e_{6,7}$	&	20	\\	
			$e_{3,6}$	&	24	  &	  $e_{7,8}$	  &	  40	\\	
			$e_{4,2}$	&	16	  &	 $e_{8,5}$	  &	30	\\	\hline
		\end{tabular}  
		\label{tbl-diffuse-one}  
	}  
	\hspace{5ex} 
	\subtable[The traffic volume $f(e)$ after the $1^{st}$ spread time window.]{    
		\vspace{1ex}       
		\begin{tabular}{cc|cc}	\hline
			\edge	&	$f(e)$	&	\edge	&	$f(e)$   \\	\hline
			$e_{1,2}$	&	10	&	$e_{4,5}$	  &	12	\\
			$e_{2,3}$	&	38	  &	  $e_{5,6}$	  &	 42	\\	
			$e_{3,4}$	&	11&	$e_{6,7}$	&	25	\\	
			$e_{3,6}$	&	26	  &	  $e_{7,8}$	  &	  30	\\	
			$e_{4,2}$	&	10	  &		$e_{8,5}$	  &	35	\\	\hline
		\end{tabular}    
		\label{tbl-diffuse-two}  
	}  
\end{table}  

%% file: sec4-algorithm.tex
\section{Our Approach}\label{sec-approach}
In this section, we propose a two-phase approach to solve the
problem -- \textit{\phaseOne} and \textit{\phaseTwo} in
Section~{\ref{sec-phase1}} and Section~{\ref{sec-phase2}},
respectively.
The first phase \textit{\phaseOne} obtains the edges influenced in
each spread time window, and then filters out the candidate edges
which violate the \duration-consecutive constraint.
The output of \textit{\phaseOne} is passed to the next phase,
\textit{\phaseTwo}, which selects the $\resultNum$ seed edges that
maximize the influence score.
We propose two different algorithmic solutions: a best-first (\bfa)
algorithm and the sampling-based greedy (\sg) algorithm with
approximation guarantees in Section~{\ref{sec-gr}} and
Section~{\ref{sec-grs}}, respectively.

\subsection{Influence Acquisition}\label{sec-phase1}
The \phaseOne mainly contains two steps as shown in
Algorithm~\ref{alg-ia}.

\input{alg-ia}

(1) Get the influenced edges (i.e., $I_j(\{e\})$) for each edge
$e$ in each spread time window $j$
(lines~\ref{algo_i-for-start}-\ref{algo_i-for-end}).
That is, for each edge $\edge \in \edgeSet$ in the $j$-th spread
time window, we first start from \edge and use depth-first search
over a road network graph \roadNetwork to obtain all the connective
neighbor edges $N(\edge)$ within \spreadTime spatial range
(line~\ref{algo_i-for-for-traverse}).
Note that, this operation can be done \textit{offline} as we know
the road length a priori and can thus estimate the travel speed.
Next, for each connective neighbor edge \edgeprime, we compare
the traffic volume $f_j(\edgeprime)$ with the congestion threshold
and determine whether it is a congested edge
(line~\ref{algo_i-for-for-for-if-start}) as defined in
Definition~\ref{def-bottleneck}.
We regard an edge as a traffic bottleneck if it can reach other roads
that are congested.
This is an \textit{online} operation since the real traffic volume
$f_j(\edgeprime)$ in the $j$-th spread time window can
continuously change over time.
Finally, we update the traffic volume for each edge based on the
Equation~\ref{equ-flow} (line~\ref{algo_i-for-for-update}), which is
used as the input to compute the next spread time window.

(2) Validate the influenced edges in $H$ (obtained in the first step)
which cannot satisfy the $\tau$-consecutive constraint (from
Definition~\ref{def-bottleneck-new})
(lines~\ref{algo_i-forfor-start}-\ref{algo_i-forfor-end}).
The intersection of influenced edges in the $\tau$ consecutive spread
time windows are computed and inserted into the influenced set $I(\{e\})$
for each edge $e$ in line~\ref{algo_i-forfor-for-intersection}.


\myparagraph{Time complexity}
In lines~\ref{algo_i-for-start}-\ref{algo_i-for-end}, the total time
complexity is $\floor*{\frac{W}{w}}|E|O(|E|)$, where $O(|E|)$ is
spent to check whether an edge is congested given the current traffic
volume.
In lines~\ref{algo_i-forfor-start}-\ref{algo_i-forfor-end}, the time
complexity is $O(|E|\tau (\floor*{\frac{W}{w}}-\tau)\Omega)$, where
$\Omega$ is the time complexity for a single intersection or union
operation between two sets.

\subsection{Bottleneck Identification}\label{sec-phase2}
Given all edges and the influenced edges, {\phaseTwo} selects a small
set of \resultNum seed edges from \edgeSet.
We introduce two different approaches to achieve the goal in Sections
\ref{sec-gr} and \ref{sec-grs}, respectively.

\subsubsection{Best-First Algorithm (\bfa)}\label{sec-gr}
The best-first algorithm \bfa follows the strategy that each time we
add an edge $e$ into the seed set \seedSet if $e$ provides the
maximum marginal gain, which is $\Phi(S\cup{\edge})-\Phi(S)$.
The \bfa algorithm has an approximation ratio guarantee of $(1-1/e)$,
given its monotonicity and submodularity properties as
described in Lemma~\ref{lemma-mono} and Lemma~\ref{lemma-sub},
respectively.

\begin{lemma}
\label{lemma-mono}
(Monotonicity) Let \seedSet and \seedSetprime be two sets of seed
edges, and $\seedSet \subseteq \seedSetprime$.
Then,
\begin{equation}
\Phi(\seedSet) \leq \Phi(\seedSetprime)
\end{equation}
\end{lemma}
\begin{proof}
Since $\seedSet \subseteq \seedSetprime$, we have
$I(\seedSetprime)=I(\seedSet) \cup I(\seedSetprime-\seedSet)$, then
$I(\seedSet) \subseteq I(\seedSetprime)$.
Therefore, $\Phi(\seedSetprime)=\sum_{e \in
I(\seedSetprime)}|I(\{e\})|$ $\geq \sum_{e \in I(\seedSet)}^{}$
$|I(\{e\})|= \Phi(\seedSet)$.
\end{proof}

\begin{lemma}
\label{lemma-sub}
(Submodularity) Let \seedSet and \seedSetprime be two sets of seed
edges, and $\seedSet \subseteq \seedSetprime$.
Assume that \seed is a newly inserted edge, we have:
\begin{equation}
\Phi(\seedSet \cup \{\seed\}) - \Phi(\seedSet) \geq \Phi(\seedSetprime \cup \{\seed\}) - \Phi(\seedSetprime)
\end{equation}
\end{lemma}
\begin{proof}
Based on differing relationships between $I(\seedSet)$ and
$I(\{\seed\})$, we analyze all possible outcomes 
as follows:

\myparagraph{Case 1}
If $I(\{\seed\}) \subseteq I(\seedSet)$, then $I(\{\seed\}) \subseteq
I(\seedSetprime)$.
Thus $\Phi(\seedSet \cup \{\seed\}) - \Phi(\seedSet)$ =
$\Phi(\seedSet) - \Phi(\seedSet) = 0$.
Similarly, $\Phi(\seedSetprime \cup \{\seed\}) - \Phi(\seedSetprime)$
= $\Phi(\seedSetprime) - \Phi(\seedSetprime) = 0$.

\myparagraph{Case 2}
If $I(\{\seed\}) \cap I(\seedSet) = \emptyset$, then $I(\{\seed\})
\cap I(\seedSetprime) = \emptyset$.
Then $\Phi(\seedSet \cup \{\seed\}) - \Phi(\seedSet)$ =
$\Phi(\seedSet) + \Phi(\{\seed\})- \Phi(\seedSet)=\Phi(\{\seed\})$.
As such, $\Phi(\seedSetprime \cup \{\seed\}) - \Phi(\seedSetprime) =
\Phi(\{\seed\})$.

\myparagraph{Case 3}
If $I(\{\seed\}) \cap I(\seedSet) \neq \emptyset$, we assume that
$A=I(\{\seed\})-I(\{\seed\}) \cap I(\seedSet)$, then $\Phi(\seedSet
\cup \{\seed\}) - \Phi(\seedSet) = \sum_{e \in A}^{}|I(\{e\})|$.
We also assume that $A^\prime=I(\{\seed\})-I(\{\seed\}) \cap
I(\seedSetprime)$, then $\Phi(\seedSetprime \cup \{\seed\}) -
\Phi(\seedSetprime) = \sum_{e \in A^\prime}^{}|I(\{e\})|$.
Note that $A^\prime \subseteq A$, thus we can easily obtain that
$\sum_{e \in A}^{}|I(\{e\})| \geq \sum_{e \in A^\prime}^{}|I(\{e\})|$
based on the monotonicity property in Lemma \ref{lemma-mono}.
Thus $\Phi(\seedSet \cup \{\seed\}) - \Phi(\seedSet) \geq
\Phi(\seedSetprime \cup \{\seed\}) - \Phi(\seedSetprime)$.

Therefore, the relationship is submodular.
\end{proof}

%
%
%
%

\subsubsection{Sampling-based Greedy Algorithm (\sg)}\label{sec-grs}
The \bfa algorithm requires $O(\resultNum |\edgeSet|)$ estimates in
total and a seed edge is selected iteratively.
However, most of these estimates are wasted since we only care
about a small set of edges with the greatest influence spread.
Therefore, to reduce the time consumption of \bfa in larger datasets,
a sampling-based greedy algorithm is proposed and referred to as
{\sg} henceforth.
The key idea is to sample a subset of edges as candidates instead of
using all the edges.
It is also possible to prove that {\sg} can maintain an equivalent
approximation guarantee to our previous approach.

Before introducing our algorithm, we first define the data structure
$RI$ which denotes the reverse influenced edges set.
An {\ri} structure contains an edge $e$ as the key and a list of edges
that impact $e$ are the values.
For example, if an edge \edge influences \edgeprime, then
$I(\{\edge\}) = \{\edgeprime\}$ and $RI(\{\edgeprime\}) = \{\edge\}$.
Now for each edge \edge, we can record its reverse influenced edges
set $RI(\{\edge\})$ from the influenced edges sets.

The intuition of the \sg algorithm is that if an edge \edge appears
in a large number of \ri sets, then it has a high probability of
influencing other edges.
Therefore, the expected influence of \edge will be large.
Based on this premise, if a set \seedSetOpt with size
\resultNum covers the most \ri sets, then \seedSetOpt will
have the maximum expected influence for all size-\resultNum edge
sets in \edgeSet.

\input{alg-sg}

Algorithm~\ref{alg-s} shows {\sg} in more detail.
Specifically, we first generate a certain number of random $RI$ sets
with size \sampleSize as candidates (line~\ref{algo_s-generate}).
Then we will choose an edge which can cover the most \ri sets
(line~\ref{algo_s-select}) and add it into the seed set \seedSet
(line~\ref{algo_s-push}).
Finally, the \ri sets that have been covered are removed by \seed
(line~\ref{algo_s-remove}).
One important question remains: How should we specify the
number of \ri sets sampled, which is \sampleSize?
Lemma~\ref{lemma-bound} provides the answer.
This however requires a few other key pieces of information.

\begin{lemma}
Let $x_i (i \in [1, |\edgeSet|])$ be a random Bernoulli variable that
equals 0 with the probability $p_i$ if $S \cap A_i = \emptyset$, and
1 with the probability $1-p_i$ otherwise.
Given a seed set $\seedSet \subseteq \edgeSet$ and a random \ri set
$A (|A| = \beta)$, then the expected influence of an arbitrary seed
set \seedSet using random \ri sets is:
\begin{equation}
\label{equ-expected}
\mathbb{E}[I(S)] = \frac{|\edgeSet|}{\beta} \cdot \mathbb{E}[\sum_{i=1}^{\beta}x_i]
\end{equation}
\end{lemma}

To ensure the estimation in Equation~\ref{equ-expected} is accurate,
$\sum_{i=1}^{\beta}x_i$ must not deviate significantly from its
expectation with an error bound $\alpha$ and confidence $1-\delta$,
meaning that Equation~\ref{equ-diff} must hold.
For example, we normally set $\alpha=0.1$ and $\delta=0.05$.
\begin{equation}
\label{equ-diff}
P[|\frac{|\edgeSet|}{\sampleSize}\sum_{i=1}^{\beta}x_i - \mathbb{E}[I(S)]| \leq \alpha \mathbb{E}[I(S)]] \geq 1 - \delta
\end{equation}

For simplicity, since $x_i \sim Bernoulli(p)$ and all related $x_i$
variables are independent, then $x \sim Binomial(|\edgeSet|, p)$,
where $p$ denotes the probability that \seedSet overlaps with a
random \ri set, and $p=\frac{\mathbb{E}[I(S)]}{|\edgeSet|}$.
Let $X_\sampleSize = \sum_{i=1}^{\beta}x_i$, we have
$\mathbb{E}[I(S)]=\mathbb{E}[x]=p |\edgeSet|$.
Then the Equation~\ref{equ-diff} can be reduced as:
\begin{equation}
	P[|\frac{|\edgeSet|}{\sampleSize} X_\sampleSize - p |\edgeSet| | \leq \alpha p |\edgeSet|] \geq 1 - \delta
\end{equation}
Then, we have:
\begin{equation}
	P[|X_\sampleSize - p \sampleSize| \leq \alpha p \sampleSize] \geq 1 - \delta
\end{equation}

\begin{lemma}
\label{lemma-bound}
Given a confidence parameter $\delta$, a sufficiently small error
parameter $\alpha$ and a sampling size $\beta=O(\frac{2 +
\alpha}{\alpha^2 p}\ln \frac{2}{\delta})$, then for any set \seedSet
with \resultNum edges, the following inequality holds with a probability
that is at least $1 - \delta$:
\begin{equation}
P[|X_\sampleSize - p \sampleSize| \leq \alpha p \sampleSize] \geq 1 - \delta
\end{equation}
\end{lemma}
\begin{proof}
Using a two-sided Chernoff bound, for any $\alpha \geq 0$,
we have
\begin{equation}
\label{equ-bound}
P[|X_\sampleSize  - p \sampleSize| \geq \alpha p \sampleSize] \geq 2 exp(-\frac{\alpha^2}{2+\alpha}\cdot  p\beta)
\end{equation}
If we want a confidence of $1-\delta$ in the estimation, we would
like the right side of the Equation~\ref{equ-bound} to be at most
$\delta$, which is:
\begin{equation}
\label{equ-size}
\delta \geq 2 exp(-\frac{\alpha^2}{2+\alpha}\cdot p \beta) \Leftrightarrow \frac{\alpha^2}{2+\alpha} \cdot p \beta \geq \ln \frac{2}{\delta} \Leftrightarrow \beta \geq \frac{2 + \alpha}{\alpha^2 p}\ln \frac{2}{\delta}. 
\end{equation}
\end{proof}
In essence, \sg will still hold the $1-1/e$ approximation guarantee
as it has the same iteration logic as the \bfa algorithm.
That is, the approximation guarantee is maintained with $(1-\delta)$
probability as only a subset of edges are considered in
\sg, when $\delta$ is a confidence parameter such as 0.05.

%% file: alg-ia.tex
\begin{algorithm}[t]
	\caption{Influence Acquisition} 
	\label{alg-ia} 
	\KwIn{A road network $\roadNetwork=\langle \vertexSet, \edgeSet \rangle$, a positive parameter \resultNum, a traffic congestion parameter \flowPara, a consecutive time interval $\tau$, a spread time window $w$, and a monitor time window $W$}
	\KwOut{The influenced edges $I(\{e\})$ for each edge $e$}
	$H \leftarrow \emptyset$, $I \leftarrow \emptyset$	\label{algo_i-init-seed}	\\	
	\For{$1 \leq j \leq \floor*{\frac{W}{w}}$}
	{	\label{algo_i-for-start}
		\For{\edge in $\edgeSet$}
		{	\label{algo_i-for-for-start}
			Traverse the road network graph \roadNetwork starting from \edge in depth-first search, and get all the connective neighbor edges $N(\edge)$ within \spreadTime 	\label{algo_i-for-for-traverse}\\
			\For{$\edgeprime$ in $N(\edge)$}
			{	\label{algo_i-for-for-for-start}
					Compute the influenced value $Inf(e, e^\prime)$ using Equation \ref{equ-influ}	\\
					\If{$Inf(e, e^\prime) \textgreater 0$ and $f_j(e^\prime) \geq \flowPara \times len(e^\prime)$}	
					{	\label{algo_i-for-for-for-if-start}
						$I_j(\{e\}).push(\edgeprime)$
					}
			}
			$H.push(\langle j, I_j(\{e\})\rangle)$	\\
			Update the traffic volume of $e$ using Equation \ref{equ-flow} \label{algo_i-for-for-update}	\\
		}	\label{algo_i-for-for-end}	
	}	\label{algo_i-for-end}
	\For{\edge in $\edgeSet$}
	{	\label{algo_i-forfor-start}
		\For{$0 \leq j \leq \floor*{\frac{W}{w}}-\tau$}
		{
			$I(\{\edge\}).push(\cap_{t=1}^{\tau}{H[j+t]} \cup I(\{\edge\}))$	\label{algo_i-forfor-for-intersection} \\ 
		}
		$I.push(<e, I(\{\edge\})>)$
	}	\label{algo_i-forfor-end}
	return $I$
\end{algorithm}

%% file: alg-sg.tex
\begin{algorithm}[h]
	\caption{Sampling-based Greedy Algorithm} 
	\label{alg-s} 
	\KwIn{The influenced edges set $I(\{e\})$ for each edge $e$, a positive parameter \resultNum}
	\KwOut{The seed edges set $\seedSet$ such that $|\seedSet|=\resultNum$}  
	Initialize a set $A=\emptyset$ \\
	Generate \sampleSize random \ri sets as candidates and insert them into $A$ \\     \label{algo_s-generate}
	\While{$|\seedSet| \textless \resultNum$}
	{    \label{algo_s-while-start}
		Select the edge $s \in \edgeSet$ that covers the most number of \ri sets in $A$  \label{algo_s-select} \\
		$\seedSet \leftarrow \seedSet \cup \{\seed\}$ \label{algo_s-push} \\
		Remove all \ri sets that are covered by $\seed$  from $A$  \label{algo_s-remove}	\\
	}
	return $\seedSet$
\end{algorithm}

%% file: sec5-experiment.tex
\section{Experiment}\label{sec-exp}
In our experimental study, we aim to investigate the following
questions.
\begin{itemize}
\item \textbf{Q1}. How sensitive are our new approaches to parameter choice,
and how do the choices affect efficiency and effectiveness trade-offs? 
\item \textbf{Q2}. How well do our methods scale as the dataset 
size increases?
\item \textbf{Q3}. How should the seed sets selected be evaluated when 
using our methods in a real traffic monitoring scenario?
\end{itemize}
\subsection{Experimental Setup}
\myparagraph{Datasets}
We use the taxi trajectory datasets \trajectorySet of \xian, \chengdu
and \porto, where the first two datasets are from Didi Chuxing GAIA
Initiative\footnote{https://gaia.didichuxing.com}, and the third one
is from a Kaggle trajectory prediction
competition\footnote{http://www.geolink.pt/ecmlpkdd2015-challenge/dataset.html}.
Table \ref{tbl-dataset} describes the statistics of the road network
and trajectory datasets.
\vspace{-2ex}
\input{tbl-dataset.tex}
\begin{itemize}
\item \textit{Road Network.}
The spatial regions of both \xian and \chengdu are located in their
respective urban areas, i.e., the regions bounded by the 2nd Ring
Road.
The road network dataset \roadNetwork of each city is obtained from
OpenStreetMap\footnote{https://www.openstreetmap.org/} based on a
bounding box, where the latitude and longitude ranges are shown in
Table \ref{tbl-dataset}.
A road may be composed of one or more road segments, where two road
segments with opposite directions form two different edges in the
road network.
Each of the two edges undergo different kinds of traffic flow
patterns.
The road networks for these three cities are shown in
Figure~\ref{fig-road}.

\item \textit{Trajectory.}
In the raw taxi trip dataset, each driver has multiple orders, where
each order includes a series of spatial locations attached with
timestamp information.
We consider each rider order as a trajectory.

\end{itemize}
\begin{figure}[h]
	\centering
	\subfigure[\chengdu]{
		\label{fig-road-cd}
		\includegraphics[width=0.29\linewidth]{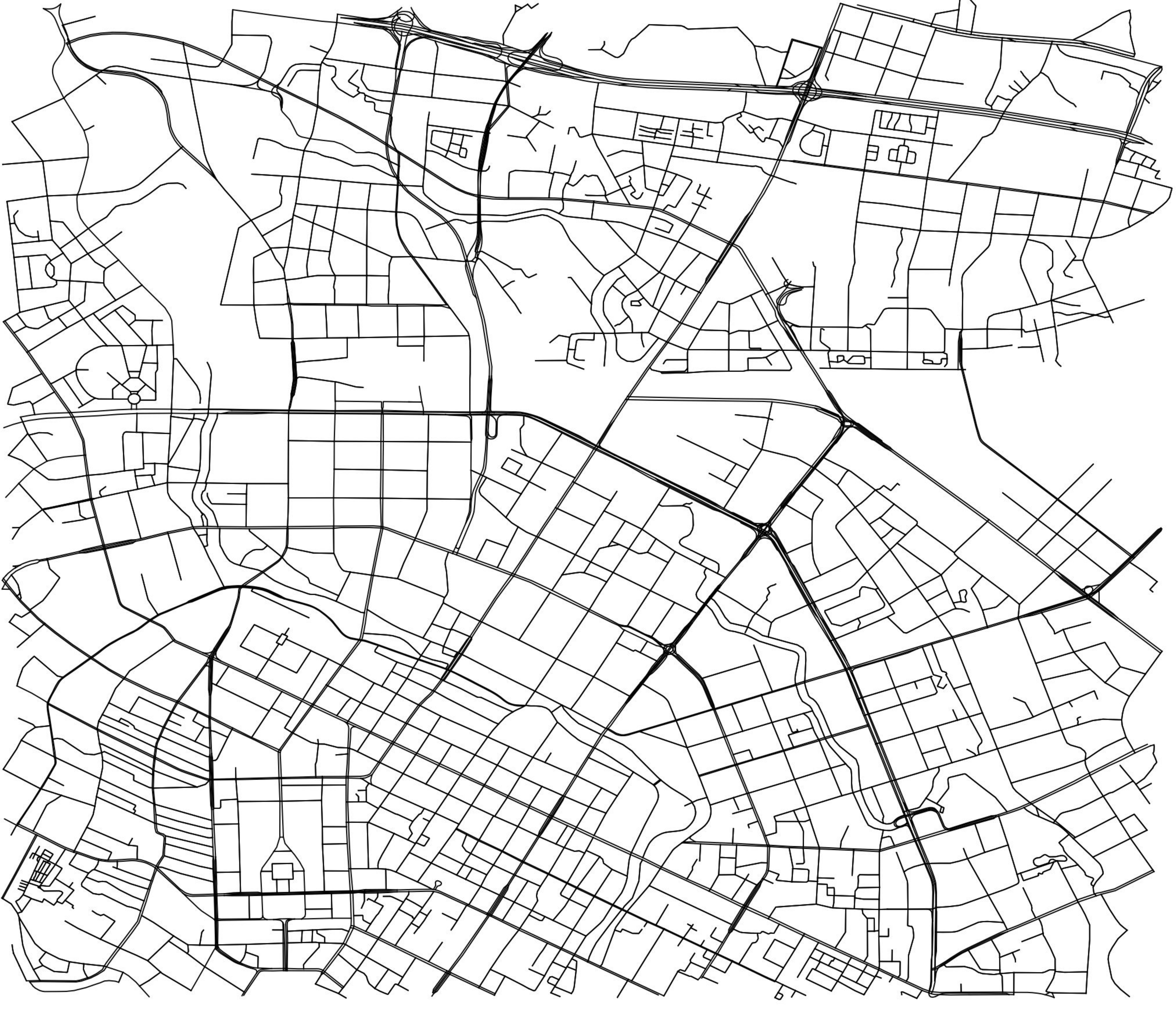}
	}
	\subfigure[\xian]{
		\label{fig-road-xa}
		\includegraphics[width=0.29\linewidth]{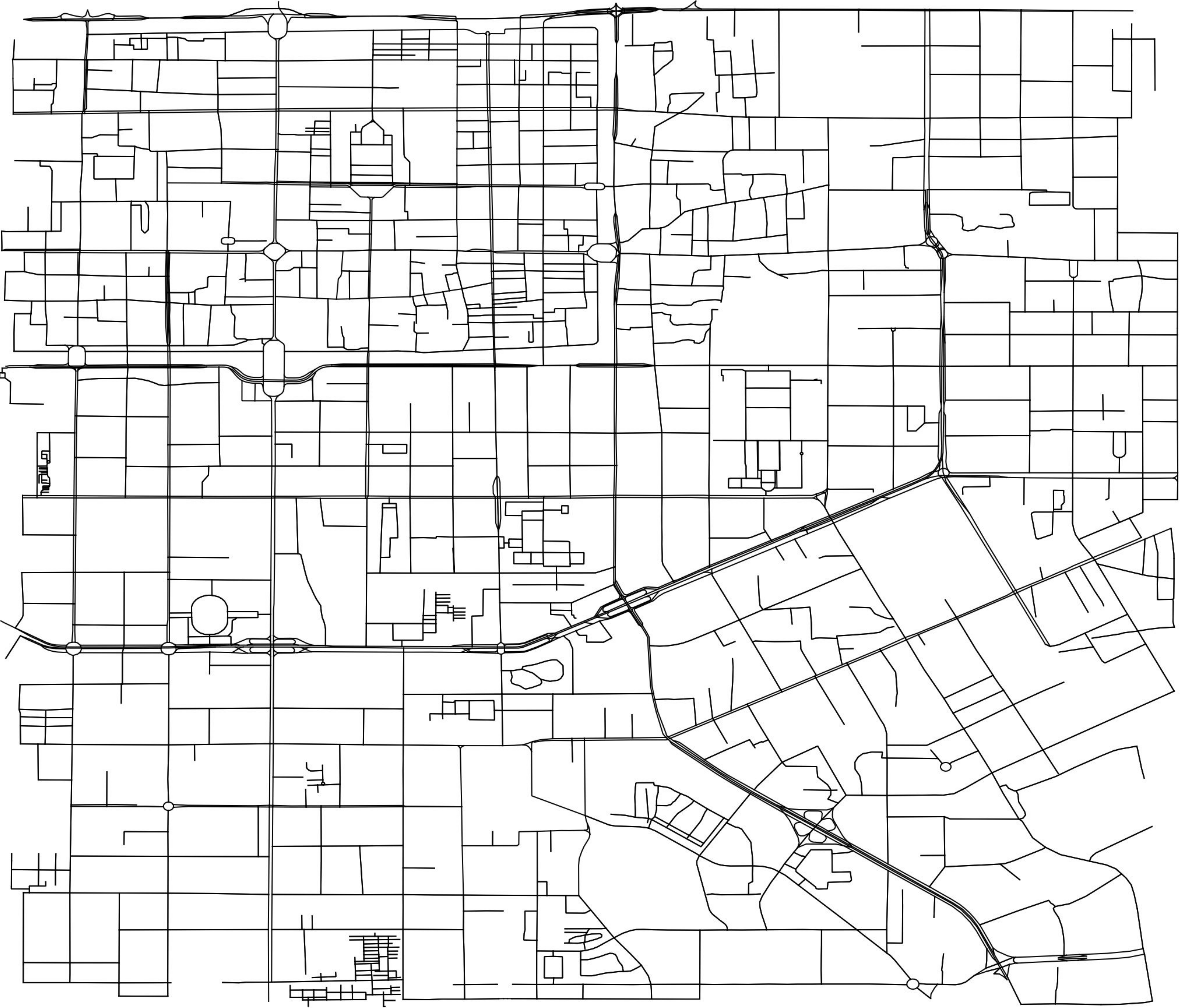}
	}
	\subfigure[\porto]{
		\label{fig-road-po}
		\includegraphics[width=0.33\linewidth]{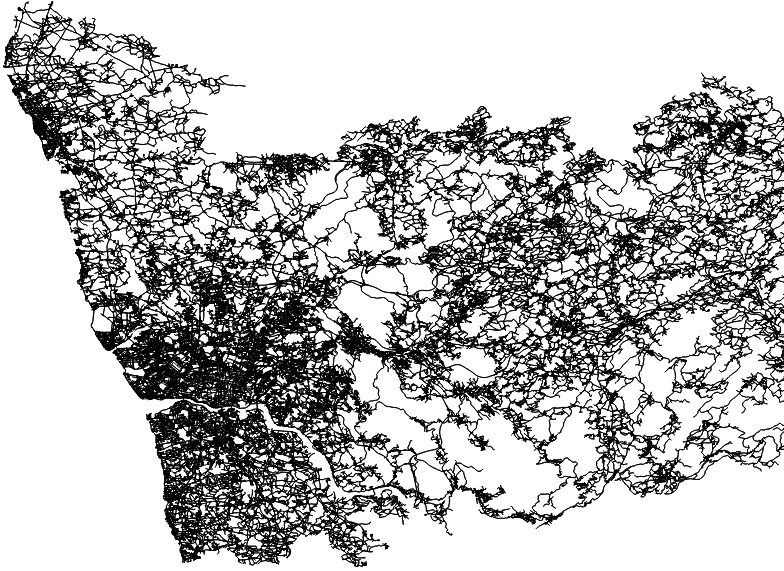}
	}
	\vspace{-2ex}
	\caption{A visualization of the three road networks used for our experiments.}
	\label{fig-road}
\end{figure}
\vspace{-2ex}
\myparagraph{Mapmatching from raw trajectories to road network}
The goal of mapmatching is to match geographic coordinates (e.g., in
the raw trajectories from vehicle GPS) in the real world to an
existing road network (i.e., a graph).
Similar to the existing work \cite{wang2019empowering}, we use the 
map matching algorithm $\mathsf{FMM}$~\cite{yang2018fast} which was
shown previously to be efficient and scalable to perform the
mapmatching and to align the raw trajectories with the corresponding
road network.
This is a \textit{one-off} and \textit{offline} operation.
Total matching times were around $3$ hours, $4.5$ hours and $50$ hours
for \xian, \chengdu and \porto, respectively.

\myparagraph{Implementation}
All experiments were performed on a server using an Intel Xeon E5 CPU
with 256 GB RAM running on Linux, implemented in C++.


\myparagraph{Algorithms}
We include several algorithms in the experimental comparison.
Since there is no previous work to solve our problem (e.g., with the
same information diffusion model), we integrate the techniques on
how to find influential edges (or nodes in different application
scenarios) and extend them to apply our traffic spread model.
\begin{itemize}
\item \topkmin algorithm, is a top-\resultNum ranking algorithm.
It first divides the whole road network based on the traffic volume
values in the road network by using the k-way partition [1], then
obtains the cut edges after dividing the road network.
Finally, we sort the cut edges in the descending order of their
average travel volume of multiple spread time windows, then generate
the top-\resultNum seed edges.
\item \cg: is a community-based greedy algorithm to find a subset of
nodes to have the maximal influence spread in a mobile social network
\cite{wang2010community}.
First, communities in a social network are detected using the
information diffusion model.
Then the algorithm applies a dynamic programming algorithm in order
to select communities and find the top-\resultNum influential nodes.
\item \cb: is a cluster-based algorithm to find a subset of trajectories
which have the maximum expected influence among a group of audiences
\cite{guo2016influence}.
First the trajectory database is divided into clusters using the
$k$-means method, and the distance between two trajectories is
computed based on overlapping POIs.
Then the algorithm locates a cluster which may generate the maximal
marginal gain, and identifies the edge which achieves that gain.
After a seed edge is selected, the marginal gain of the remaining
edges is updated, and the process continues until the number of
seed edges specified have been returned.
\item \bfa: is the best first algorithm introduced in Section~\ref{sec-gr}.
\item \sg: is the sampling-based greedy algorithm introduced in
Section~\ref{sec-grs}.
\end{itemize}

\myparagraph{Performance Measurement}
We perform both efficiency and effectiveness evaluations for all 
methods.
For the efficiency, we report the runtime to select the
top-\resultNum seed edges.
Each experiment is repeated $10$ times and the average runtime is
reported.
For the effectiveness, we show the \textit{coverage ratio} which is
computed as the influence score $\Phi(S)$ of the selected seed edges
\seedSet divided by the number of edges covered by the trajectories.
A larger coverage ratio indicates a better selection of the seed
edges.
 
\myparagraph{Parameter Setting}
Parameter settings are shown in Table~\ref{tbl-parameter}, with 
the default values shown in bold.
Instead of adopting a fixed size of \resultNum, we use an additional
parameter \resultNumPara to control the ratio on how many seed edges
(over the total number of edges) to be selected algorithmically.
Thus, we have $\resultNum = \resultNumPara \times |\edgeSet|$.
\input{tbl-parameter.tex}

\subsection{A Statistical Analysis of the Datasets}
\label{sec-stats}
We now perform a statistical comparison to support the default parameter choices as shown in
Table~\ref{tbl-parameter}.

\myparagraph{The distribution of edge length on road network}
First, we perform a statistical analysis on the edge length
distribution, which is shown in Figure~\ref{fig-road-dis}.
For example, there are $20.5\%$ edges whose lengths are smaller than
$50$ meters in the road network of \xian.
For \chengdu, it can be observed that most of the edges are less than
$450$ meters, and nearly half of the edges are within $200$ meters.
In the existing work on traffic diffusion~\cite{li2018diffusion,
yu2018spatio}, the assumption is that the traffic flow can spread at
most five hops from a certain node or edge, after which the diffusion
power will be mitigated.
Instead of using the concept of ``hop'' as the traffic spread unit,
we use a time window \spreadTime to control the range of traffic
diffusion.
Our motivation is that the real traffic speed may vary for different
edges.
In our experiments, the traffic speed was estimated by using the
historical trajectory datasets.
In particular, we calculate the traffic speed per hour.
First, we record all of the trajectories passing through each edge,
and then obtain the traffic speed computed as the distance between
two POIs in a trajectory divided by the timestamp difference.
We then use this average traffic speed in our experiments.
\begin{figure}[h]
	\centering
	\subfigure[\xian]{
		\label{fig-road-dis-xa}
		\includegraphics[width=\graphscaleThree\linewidth]{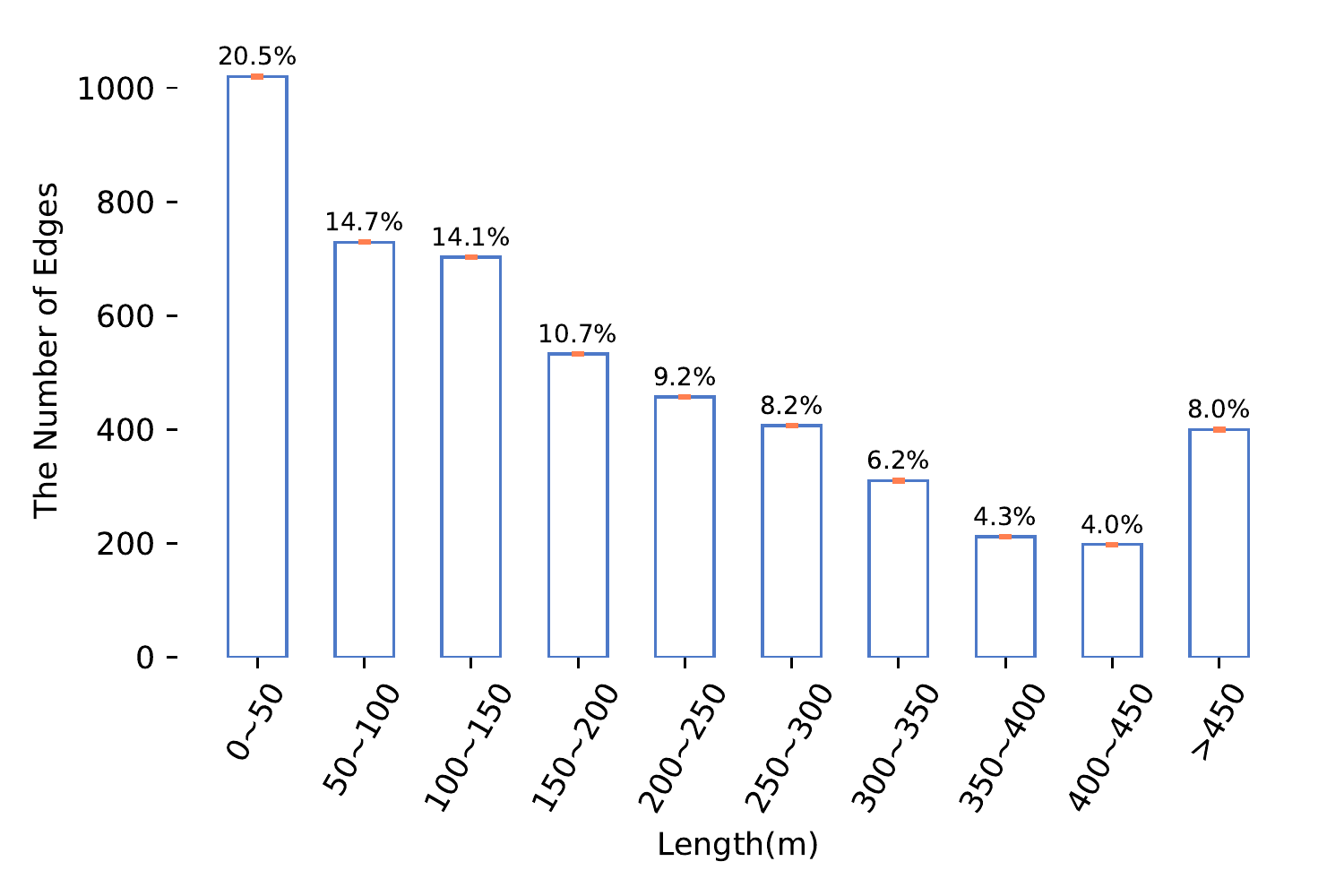}
	}
	\subfigure[\chengdu]{
		\label{fig-road-dis-cd}
		\includegraphics[width=\graphscaleThree\linewidth]{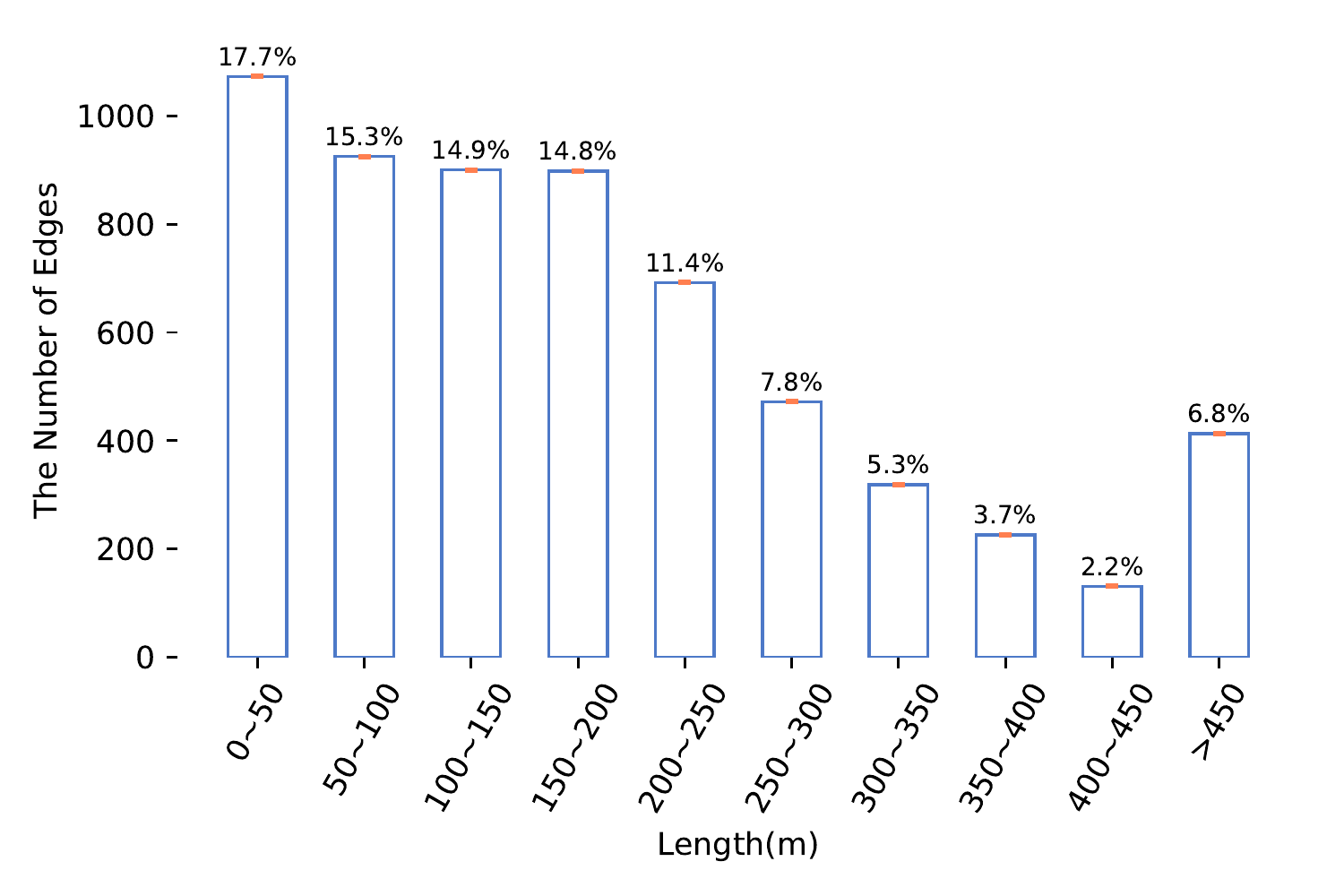}
	}
	\subfigure[\porto]{
		\label{fig-road-dis-porto}
		\includegraphics[width=\graphscaleThree\linewidth]{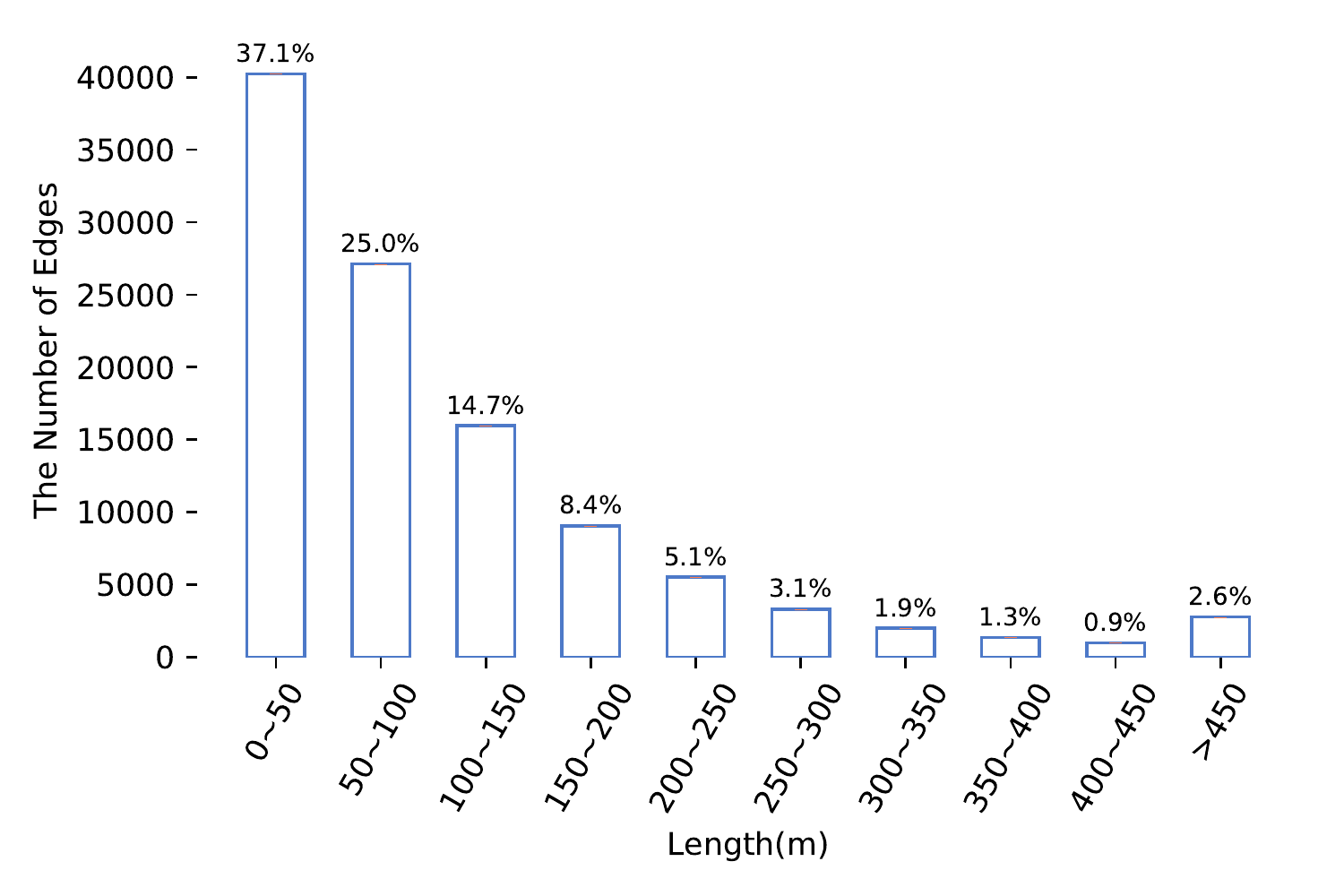}
	}
	\vspace{-2ex}
	\caption{The edge length distribution for three different road networks. }
	\label{fig-road-dis}
\end{figure}

\myparagraph{The distribution of edge volume on road network}
The statistical analysis of edge volume distributions is
illustrated in Figure~\ref{fig-road-density}, for one peak hour
(i.e., 18:00 pm-19:00 pm).
Note that, the statistical ratio is based only on edges that are
covered by trajectories.
For instance, $33.7\%$ edges have a traffic volume of less than $50$
in the road network of \xian, and $28.1\%$ edges have an edge density
of larger than $500$.
But as one might expect, there are some edges which are not covered
by any trajectory.
Specifically, the proportions of edges without covering trajectories
are $24.2\%$, $25.2\%$ and $79.1\%$ for \xian, \chengdu and \porto,
respectively.
Observe that the traffic volumes in \porto are much smaller, and the
cause is two-fold: (1) The \porto dataset used a lower
sampling rate (i.e., $15$ seconds shown in Table~\ref{tbl-dataset})
than the other test collections, which were around $2$-$4$ seconds.
(2) Even though there are more trajectories in \porto than in \xian
and \chengdu, the time span covered by the trajectories is much
longer (nearly a year) while the other two cities are for a single
day.
\begin{figure}[h]
	\centering
	\subfigure[\xian]{
		\label{fig-road-density-xa}
		\includegraphics[width=\graphscaleThree\linewidth]{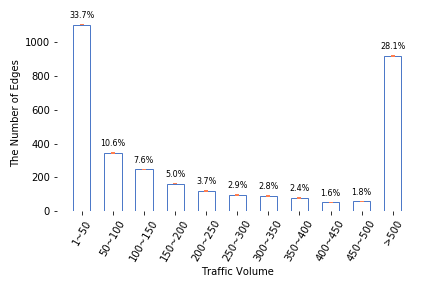}
	}
	\subfigure[\chengdu]{
		\label{fig-road-density-cd}
		\includegraphics[width=\graphscaleThree\linewidth]{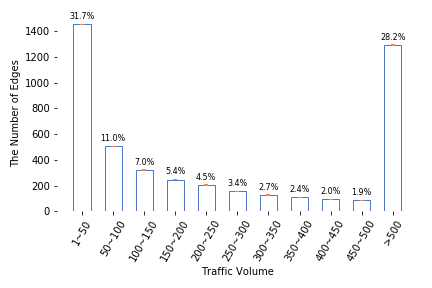}
	}
	\subfigure[\porto]{
		\label{fig-road-density-porto}
		\includegraphics[width=\graphscaleThree\linewidth]{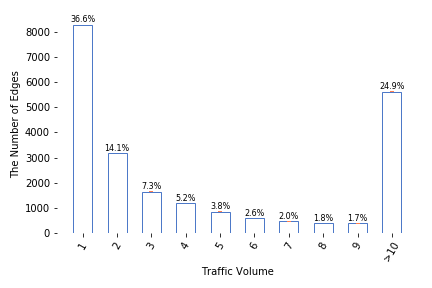}
	}
	\vspace{-2ex}
	\caption{Edge traffic volume distribution from 18:00 pm to 19:00 pm (during one peak hour) on all three road networks. }
	\label{fig-road-density}
\end{figure}
\vspace{-2ex}
\subsection{Experimental Result (Q1 and Q2)}
\subsubsection{Performance Evaluation on Influence Acquisition}
For each edge, the average runtime to identify the influenced edges
in each spread time window are $0.019$ ms, $0.019$ ms, and $0.018$ ms
for \xian, \chengdu, and \porto, respectively.
After identifying the influenced edges in each spread time window, we
check the \duration-consecutive constraint.
So, we also studied the efficiency costs to check the constraint for 
different \duration in all three datasets, which is summarized in
Table~\ref{tbl-tau}.
As \duration increases, the runtime also increases for all datasets
as more consecutive spread time window combinations have to be
validated.
Observe that the runtime of the \porto dataset is much smaller
than that of the other two datasets.
As discussed in Section~\ref{sec-stats}, a large number of edges are
not covered with trajectories in \porto dataset.
When the traffic flow spread starting from an edge \edge to another
edge \edgeprime, the traffic influenced value $Inf(\edge,
\edgeprime)$ (in Equation \ref{equ-influ}) is normally zero if
\edgeprime is not covered with trajectories as the traffic diffusion
probability $\diffusePro(e, e^\prime)$ is zero.
In such cases, the number of influenced edges is smaller, which
translates to  less runtime during \phaseOne processing.

\input{tbl-exp-tau}

\input{fig-exp-xian}
\input{fig-exp-chengdu}
\input{fig-exp-porto}
\subsubsection{Performance Evaluation for Bottleneck Identification}
In this section, we conduct an experimental study to evaluate the
efficiency and effectiveness of \phaseTwo against the baselines under
different parameter settings, over all the three datasets.
We plot multiple \textit{efficiency} and \textit{effectiveness}
trade-off graphs in order to better understand the performance
differences among all the algorithms.
The starting and end sweep values are shown for each line to make it
easier to observe the performance trends for each algorithm.

\myparagraph{Effect of the ratio of seed edges \resultNumPara}
The effect of \resultNumPara, which controls the number of selected
seed edges, on the performance is presented in
Figure~\ref{fig:xa:epsilon}, Figure~\ref{fig:cd:epsilon} and
Figure~\ref{fig:po:epsilon}.
As \resultNumPara increases, both coverage ratio and runtime
of all the algorithms increase, and more seed edges are 
selected so that more influenced edges can be covered.
Even though the \topkmin algorithm is more efficient than the other
four methods, the coverage ratio is significantly worse as it only sorts
the cut edges in the descending order of their traffic volume.
This is valuable evidence that roads (e.g., highways) with high traffic
volumes do not necessarily have the most influence on other
road segments.
Furthermore, we can observe that both \bfa and \sg consistently
outperform \cb and \cg by an order of magnitude in runtime, which
shows that our proposed methods are also efficient. 
Compared with \cg, \cb is more stable with a larger \resultNumPara
because \cg first locates a cluster with the largest estimated
marginal gain and then selects a suitable seed edge with the maximum
marginal gain.
When \resultNumPara is large, \cg has to compute the estimated
marginal gain values for more clusters.
In the \porto dataset, both \bfa and \sg can outperform \cg by two
orders of magnitude, showing that our algorithms are highly scalable.
In terms of effectiveness, \bfa consistently outperforms the other
algorithms as it can always leverage the best solution currently
found directly.

\myparagraph{Effect of the traffic congestion threshold \flowPara}
The effect of \flowPara, which constrains the traffic congestion
threshold of road congestion, on the performance is shown in
Figure~\ref{fig:xa:lambda}, Figure~\ref{fig:cd:lambda} and
Figure~\ref{fig:po:lambda}.
With a larger \flowPara, all the algorithms have a better running
time and coverage ratio since fewer edges are considered as
\textit{congested}.
In term of efficiency, \bfa and \sg still have lower runtimes than
the baselines.
For the effectiveness, \bfa consistently achieves the best coverage
ratio.

\input{tbl-exp-sampling}
\myparagraph{Effect of the traffic spread time window \spreadTime}
The effect of \spreadTime, which measures the traffic diffusion time
window, on the performance is depicted in Figure~\ref{fig:xa:spread},
Figure~\ref{fig:cd:spread} and Figure~\ref{fig:po:spread}.
When we have a larger \spreadTime, all the algorithms maintain a
stable trend for runtime and coverage ratio as more edges can
be influenced by a certain seed edge after the traffic diffusion.

\myparagraph{Effect of  sampling size}
We also carry out an experimental study on \sg when the sampling size is varied.
The performance of \sg algorithm is closely related to the sampling
size, which is the number of samples used as candidates.
When we compare \sg using different sampling sizes with \bfa as a
baseline, the result is presented in  Table~\ref{tbl-sampling}.
As one might expect, a larger sample size leads to a better coverage
ratio at the cost of longer runtime.
Based on this study, we set the default sampling size to
30\%, 30\% and 40\% for \xian, \chengdu and \porto, respectively as
it provides a competitive trade-off between efficiency and
effectiveness.

\myparagraph{Scalability Test}
When we compare the performance of all methods based on three
datasets, we find that the efficiency of both \cb and \cg increases
significantly as the collection size increases.
The time complexity of these two algorithms are proportional to the
size of clusters or communities derived from dividing the road
network.
For \cb, it has to find a suitable cluster first for every 
seed edge selected; then, it has to update the estimated
maximum marginal gain for both clusters and edges.
As \cg is based on dynamic programming
over the communities, increasing the number of communities leads to
reduced efficiency.
Larger datasets have more communities, which exacerbate the
efficiency further.
In contrast, our algorithms \bfa and \sg consistently maintain a
reasonable growth, which corresponds to the selection over edges
directly rather using clusters or communities.

\input{fig-exp-vis}
\vspace{-2ex}
\subsection{Case Study (Q3)}
The effectiveness of traffic bottleneck identification on road
networks can be demonstrated using a case study visualization with
the Plotly API\footnote{https://plotly.com/python/maps/} for the
\xian and \chengdu datasets.
We omit the visualization of the \porto dataset since the spatial
region is much larger resulting in sparsity of covered edges.
The case study uses data from a typical  peak rush hour (from
18:00 pm to 19:00 pm).
We hope to answer two questions using this visualization: (1) How do
the congested road segments map to a road network in a real scenario;
(2) How are the selected seed edges (which are considered as traffic
bottlenecks) influencing the other edges?

We first show congested roads which are highlighted in red in
Figure~\ref{fig:vis:xa:cong} and Figure~\ref{fig:vis:cd:cong}.
Then in the remaining figures, we present the selected seed edges and
their corresponding influenced edges.
For better visualization, we have selected $15$ seed edges, with two
endpoints plotted as circle markers and the influenced edges
displayed as black lines.
We have regarded the \bfa algorithm as a baseline, and plot also
any edges which were not influenced by each method as red lines.
More red lines implies a larger disagreement with the best solution.
The general trend appears to be that $\bfa \textgreater \sg
\textgreater \cg \textgreater \cb \textgreater \topkmin$ in terms 
of effectiveness.

%% file: tbl-dataset.tex
\begin{table}[h]
	\small
	\renewcommand{\arraystretch}{1.2}
	\caption{Statistics of road network and trajectory datasets}
	\label{tbl-dataset}
	\vspace{-1ex}
	\centering
	\begin{tabular}{lllll}
		\hline		&  &\xian	&\chengdu  	&Porto  \\ \hline
		\multirow{5}{*}{\roadNetwork}
		&\#nodes	&2,086	&4,326	&60,287	\\
		&\#edges	&5,045	&6,135	&108,571	\\
		&avg edge length (m)	&199	&194	&114	\\
		&latitude range	&[$30.6528^\circ$, $30.7278^\circ$]	&[$30.5670^\circ$, $30.7879^\circ$]	&[$40.9000^\circ$, $41.4200^\circ$]	\\
		&longitude range	&[$104.0421^\circ$, $104.1296^\circ$]	&[$103.9279^\circ$, $104.2081^\circ$]	&[$-8.7857^\circ$, $-8.2001^\circ$]	\\
		\hline
		\multirow{4}{*}{\trajectorySet} 
		&\#trajectories	&119,019	&192,901	&1,565,595	\\
		&\#points	&28,327,565	&41,664,011	&100,995,114	\\
		&time span	&Oct 1, 2016	&Oct 1, 2016	&July 1, 2013 - June 30, 2014	\\
		&sampling time (s)	&2-4	&2-4	&15	\\
		\hline
	\end{tabular}
\end{table}

%% file: tbl-parameter.tex
\begin{table}[h]
	\small
	\renewcommand{\arraystretch}{1.2}
	\caption{Parameter settings. }
	\centering
		\begin{tabular}{cc}
			\hline
			{\bfseries Parameter} & {\bf Setting}\\ 
			\hline
			$\resultNumPara (\resultNum = \resultNumPara \times |\edgeSet|)$  &	 \tabincell{c}{\xian \& \chengdu: 0.2\%, 0.6\%, \textbf{1\%}, 1.4\%, 1.8\% \\ \porto: \textbf{0.1\%}, 0.2\%, 0.3\%, 0.4\%, 0.5\%}	\\	\hline
			\flowPara (vehicle per meter)  & \textbf{1}, 2, 3, 4, 5\\ \hline
			$\spreadTime (seconds)$  &  10, \textbf{20}, 30, 40, 50\\ \hline
			$\monitorTime (seconds)$  &\textbf{3600} (i.e., 18:00 pm-19:00 pm)\\ \hline
			$\duration(\times w)$  & 1, 2, \textbf{3}, 4, 5\\ \hline
			The sampling size \sampleSize  & \tabincell{c}{\xian \& \chengdu: 20\% , \textbf{30\% }, 40\% \\ \porto: 30\% , \textbf{40\% }, 50\%}	\\ \hline
		\end{tabular}
	\label{tbl-parameter}
\end{table}

%% file: tbl-exp-tau.tex
\begin{table}[t]
	\small
	\renewcommand{\arraystretch}{1.2}
	\caption{Efficiency Study on checking \duration-consecutive constraint with varying \duration in different datasets}
	\label{tbl-tau}
	\vspace{-2ex}
	\centering
	\begin{tabular}{ccccccc}
		\duration		&Dataset      & Runtime (ms) 	&Dataset        & Runtime (ms)	&Dataset       & Runtime (ms) \\ \hline
		1&\multirow{5}{*}{\xian} &0.611 &\multirow{5}{*}{\chengdu} &0.686& \multirow{5}{*}{\porto}&	0.229                 \\
		2	&	&0.968	&		&1.077	&	&0.395                   \\
		3	&  &1.327	&		&1.450	&	&0.559                  \\			
		4	&  &1.654	&		&1.829&	&0.720                  \\			
		5	&  &1.998	&		&2.208	&	&0.883                   \\			
		\hline
	\end{tabular}
\end{table}

%% file: fig-exp-xian.tex
\begin{figure*}[h]
	\centering
	\subfigure[Varying  \resultNumPara]{
		\label{fig:xa:epsilon}
		\includegraphics[width=\graphscaleThree\textwidth]{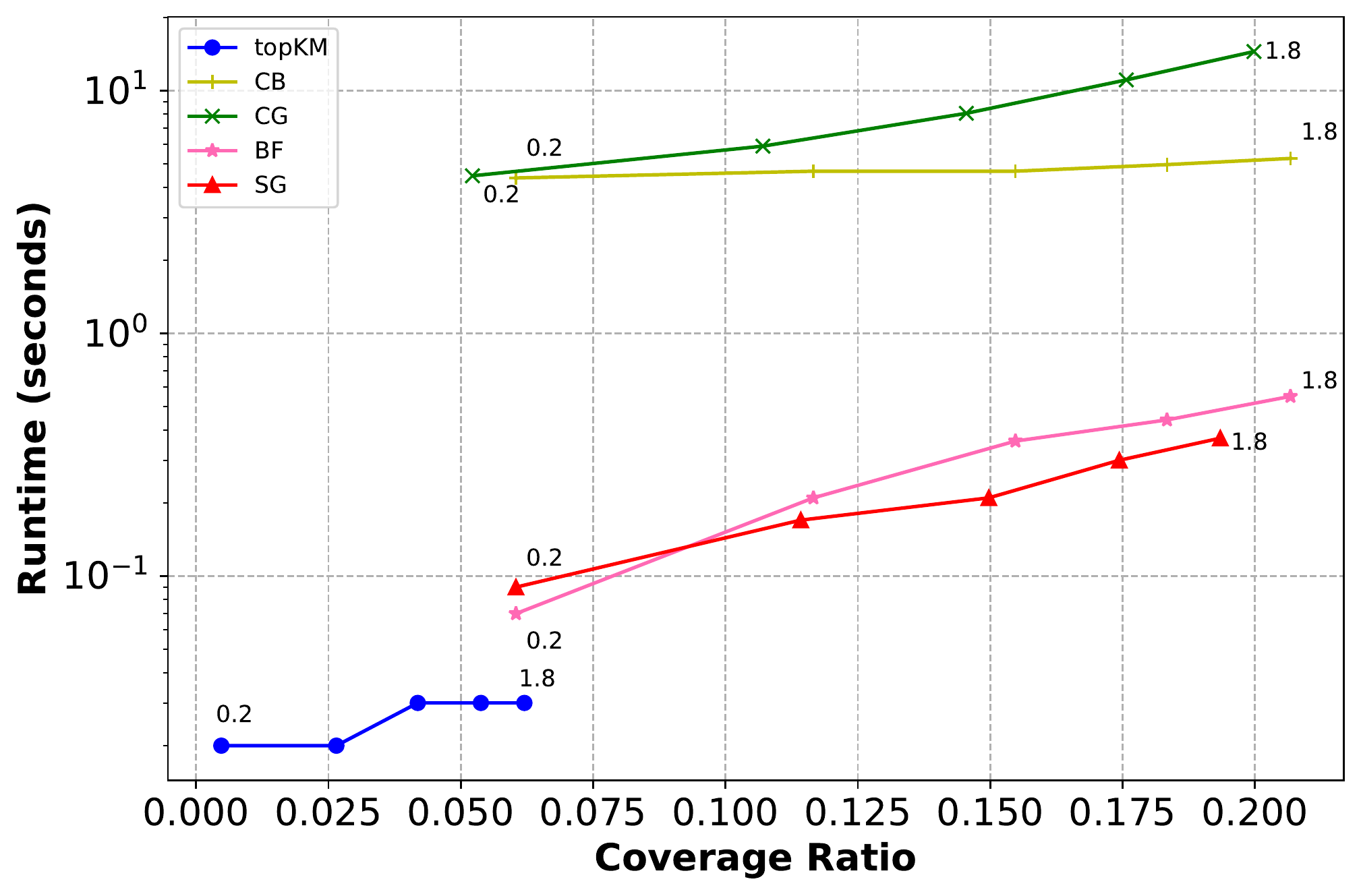}
	}
	\subfigure[Varying  \flowPara]{
		\label{fig:xa:lambda}
		\includegraphics[width=\graphscaleThree\textwidth]{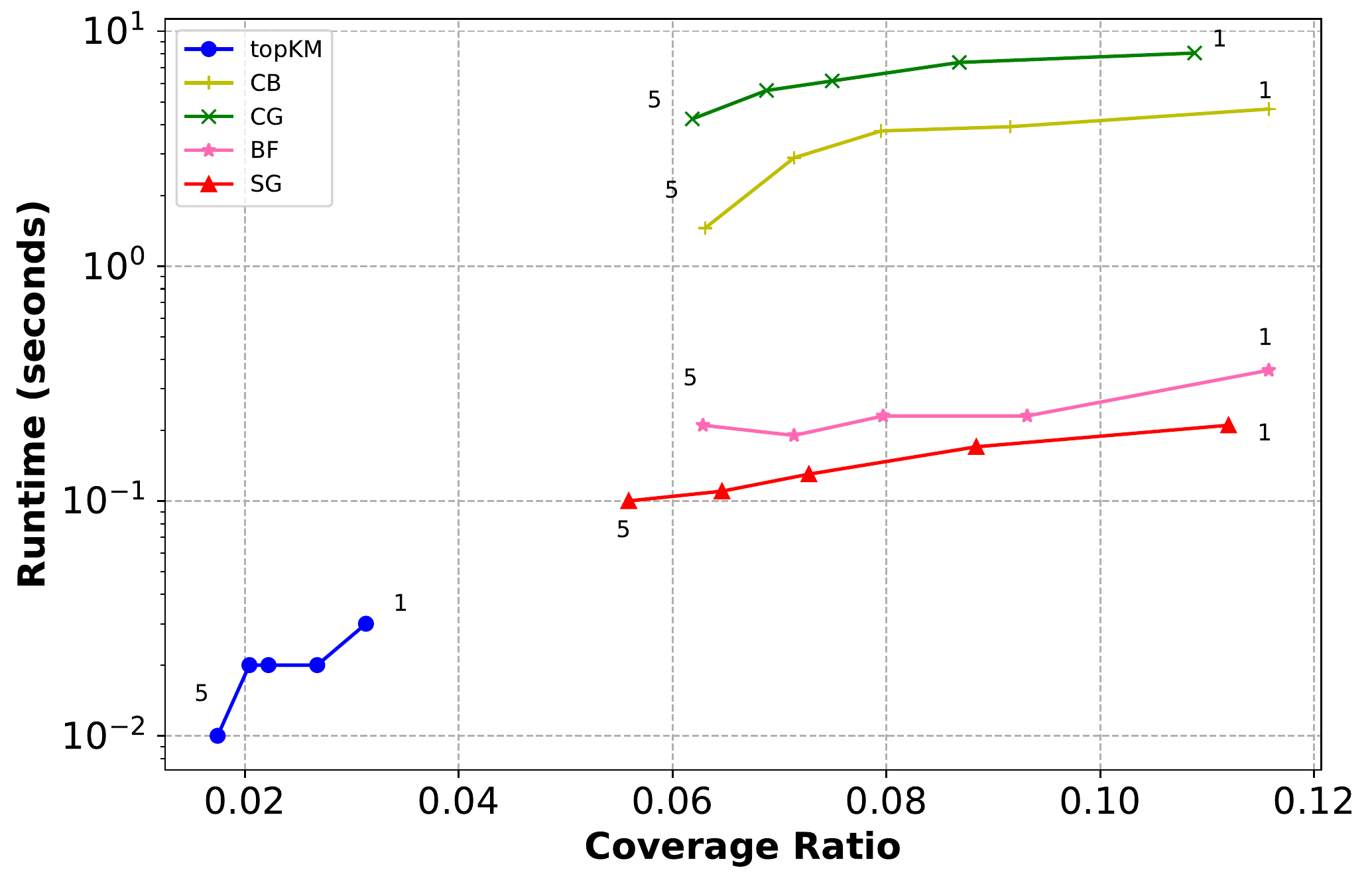}
	}
	\subfigure[Varying \spreadTime]{
		\label{fig:xa:spread}
		\includegraphics[width=\graphscaleThree\textwidth]{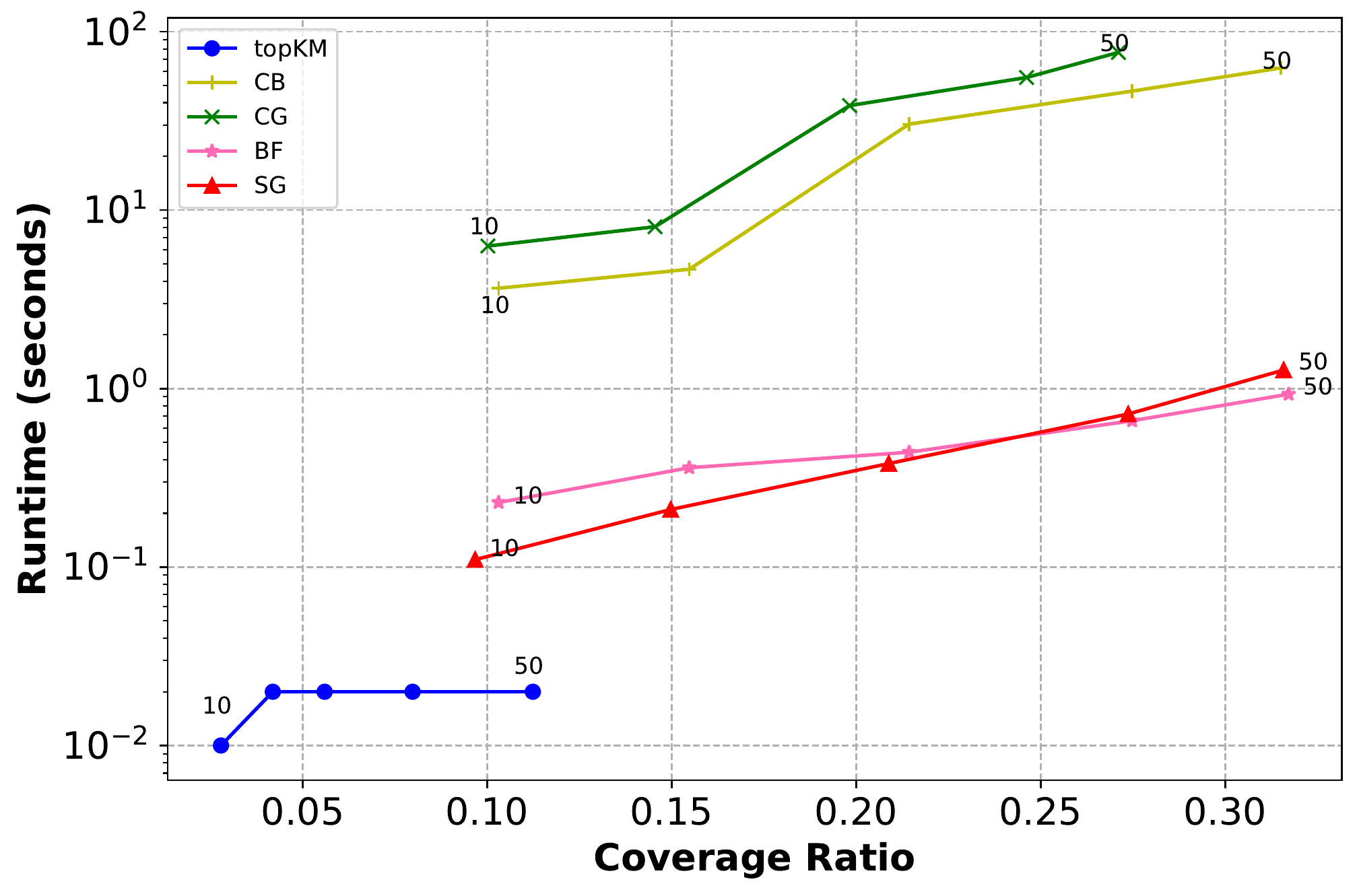}
	}
	\caption{Performance Comparison of Methods in \xian dataset.}
	\label{fig:xian}
\end{figure*}

%% file: fig-exp-chengdu.tex
\begin{figure*}[h]
	\centering
	\subfigure[Varying  \resultNumPara]{
		\label{fig:cd:epsilon}
		\includegraphics[width=\graphscaleThree\textwidth]{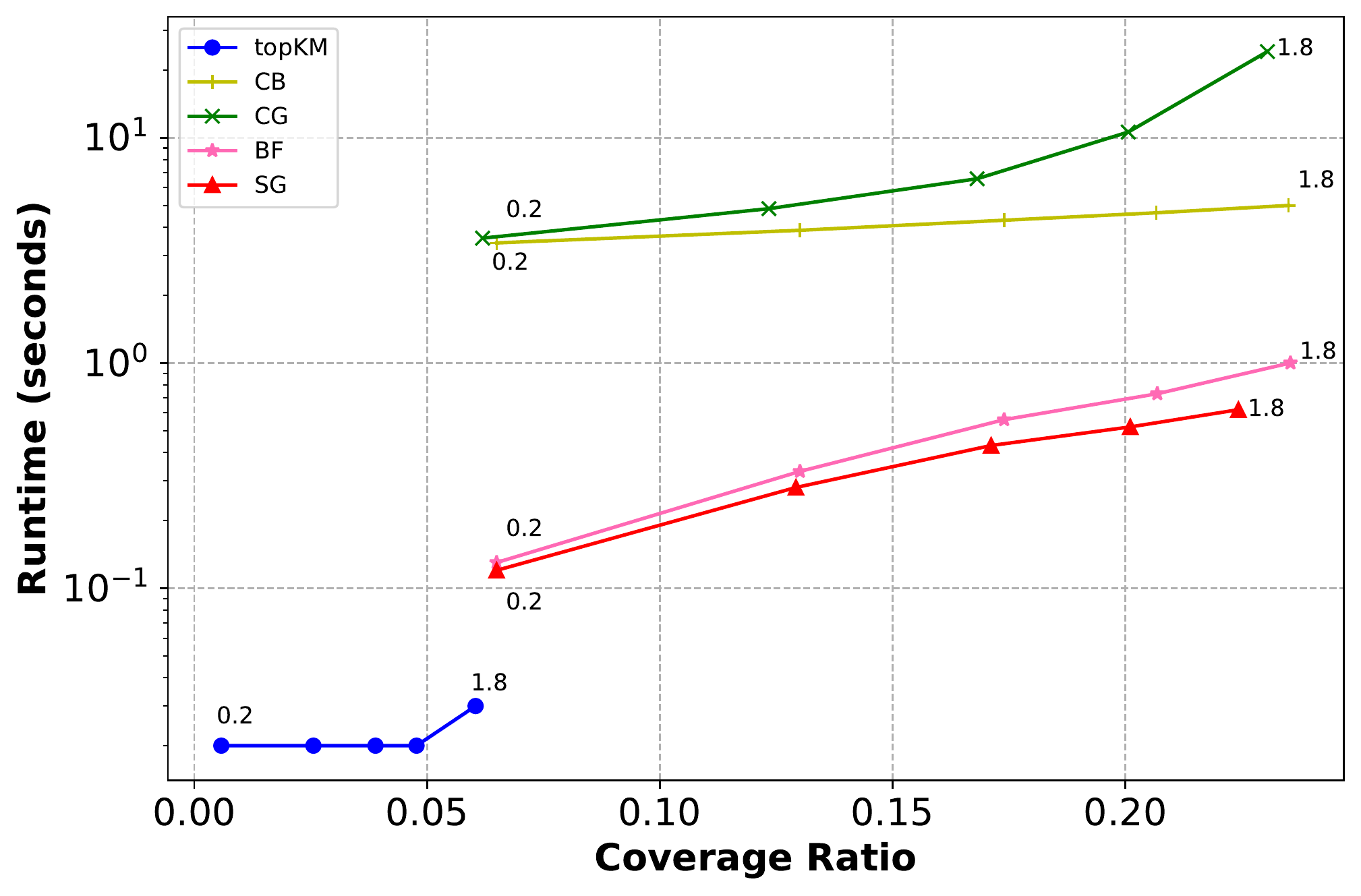}
	}
	\subfigure[Varying \flowPara]{
		\label{fig:cd:lambda}
		\includegraphics[width=\graphscaleThree\textwidth]{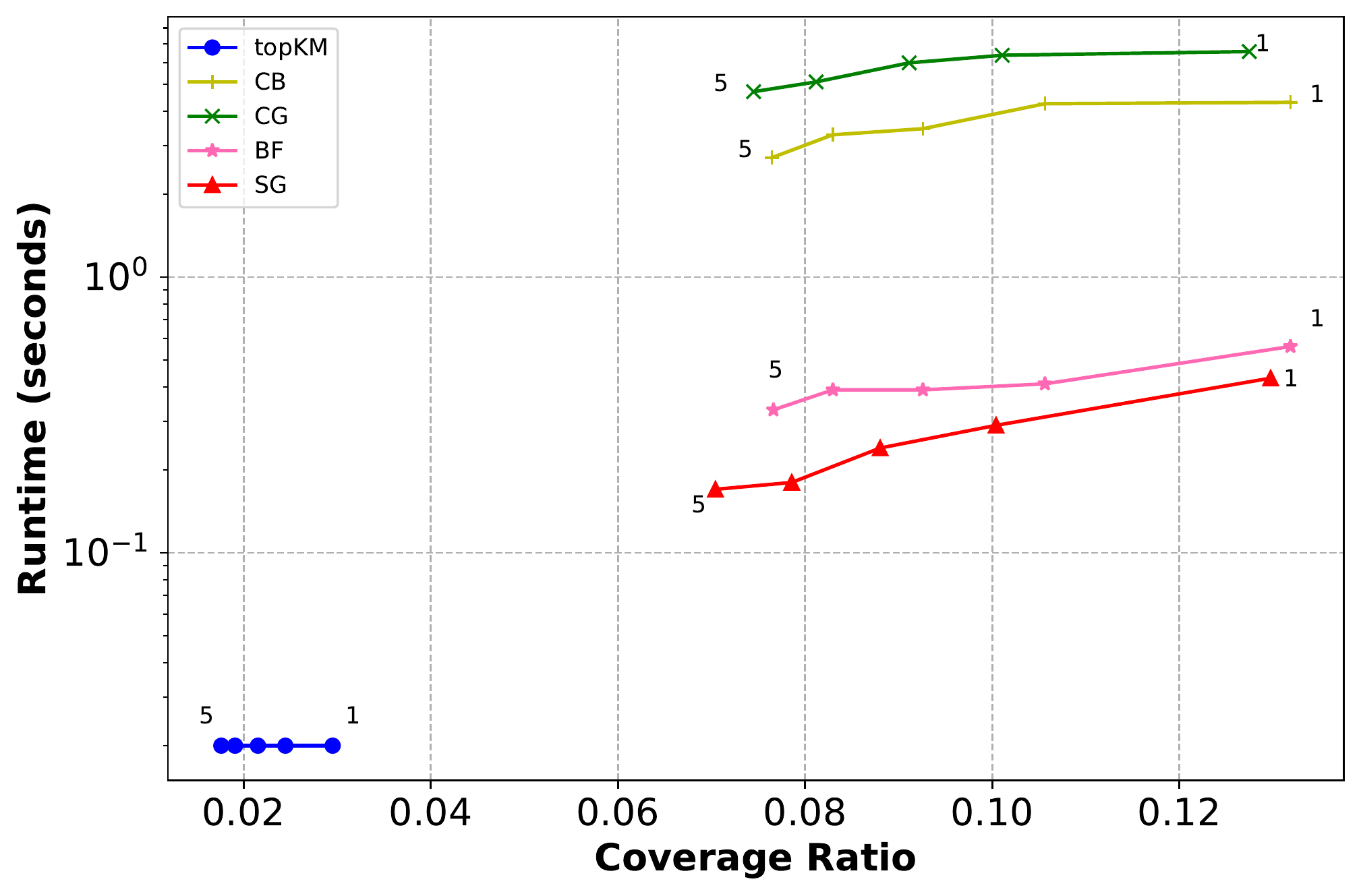}
	}
	\subfigure[Varying \spreadTime]{
		\label{fig:cd:spread}
		\includegraphics[width=\graphscaleThree\textwidth]{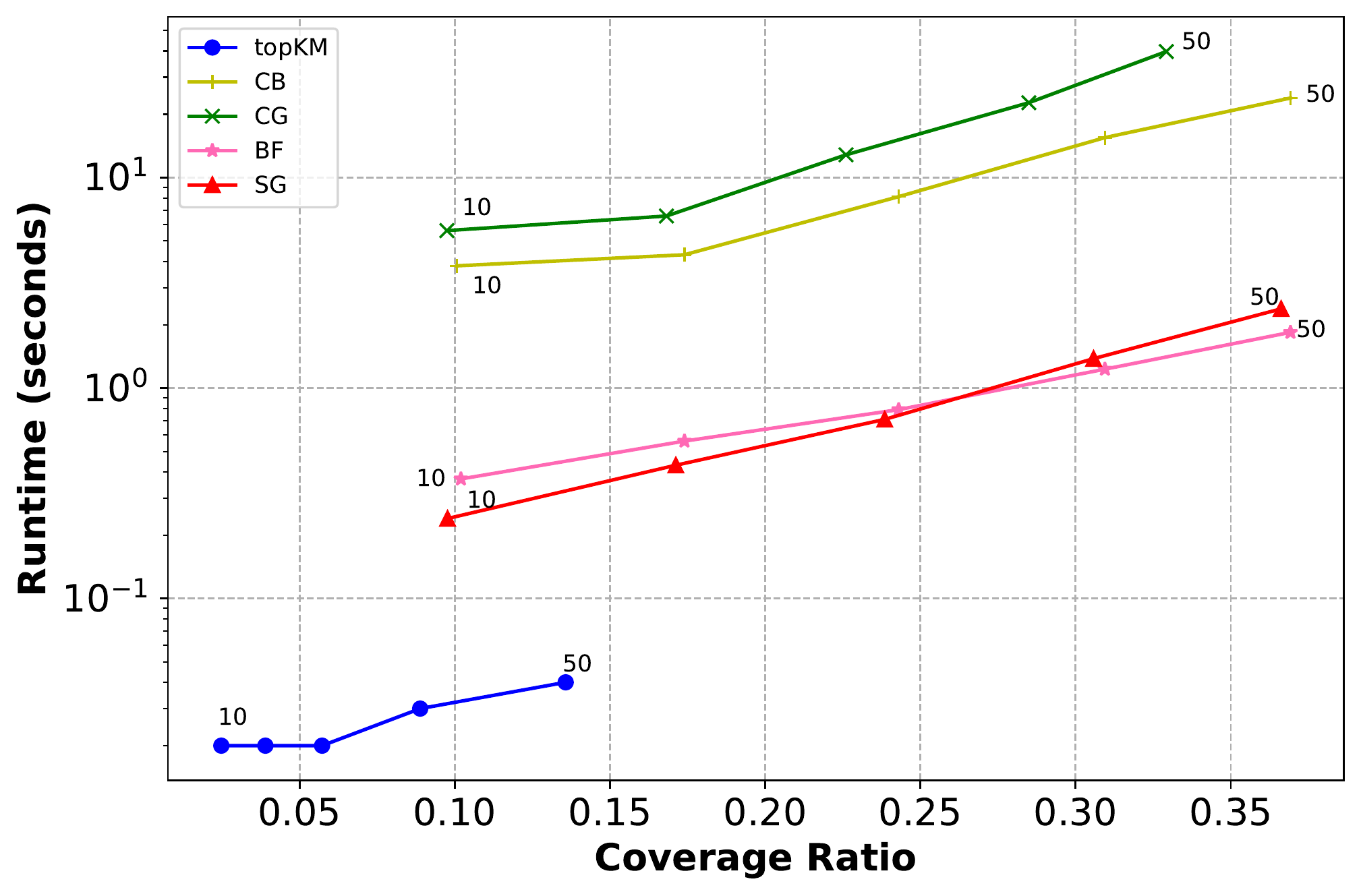}
	}
	\caption{Performance Comparison of Methods in \chengdu dataset. }
	\label{fig:chengdu}
\end{figure*}

%% file: fig-exp-porto.tex
\begin{figure*}[h]
	\centering
	\subfigure[Varying  \resultNumPara]{
		\label{fig:po:epsilon}
		\includegraphics[width=\graphscaleThree\textwidth]{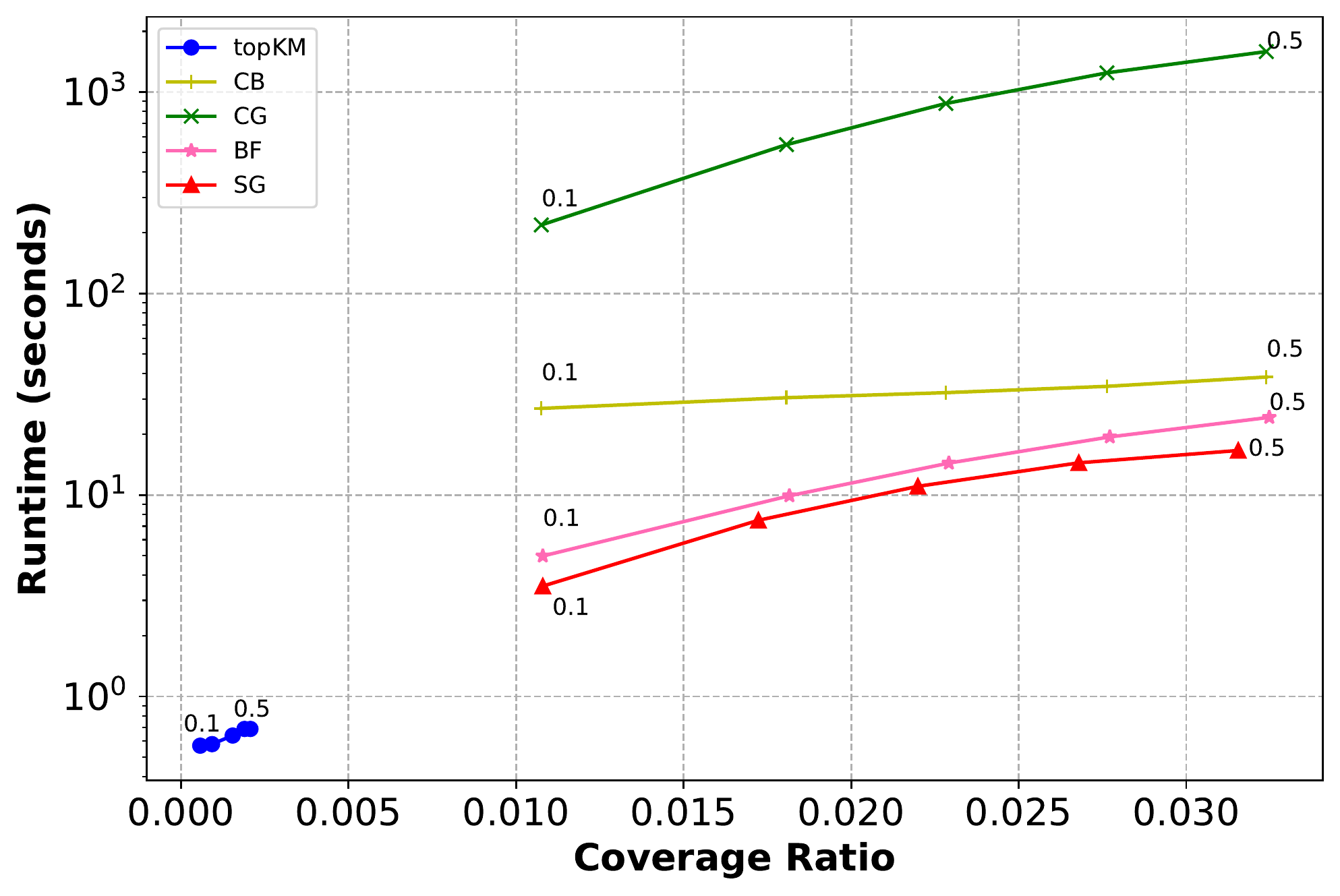}
	}
	\subfigure[Varying  \flowPara]{
		\label{fig:po:lambda}
		\includegraphics[width=\graphscaleThree\textwidth]{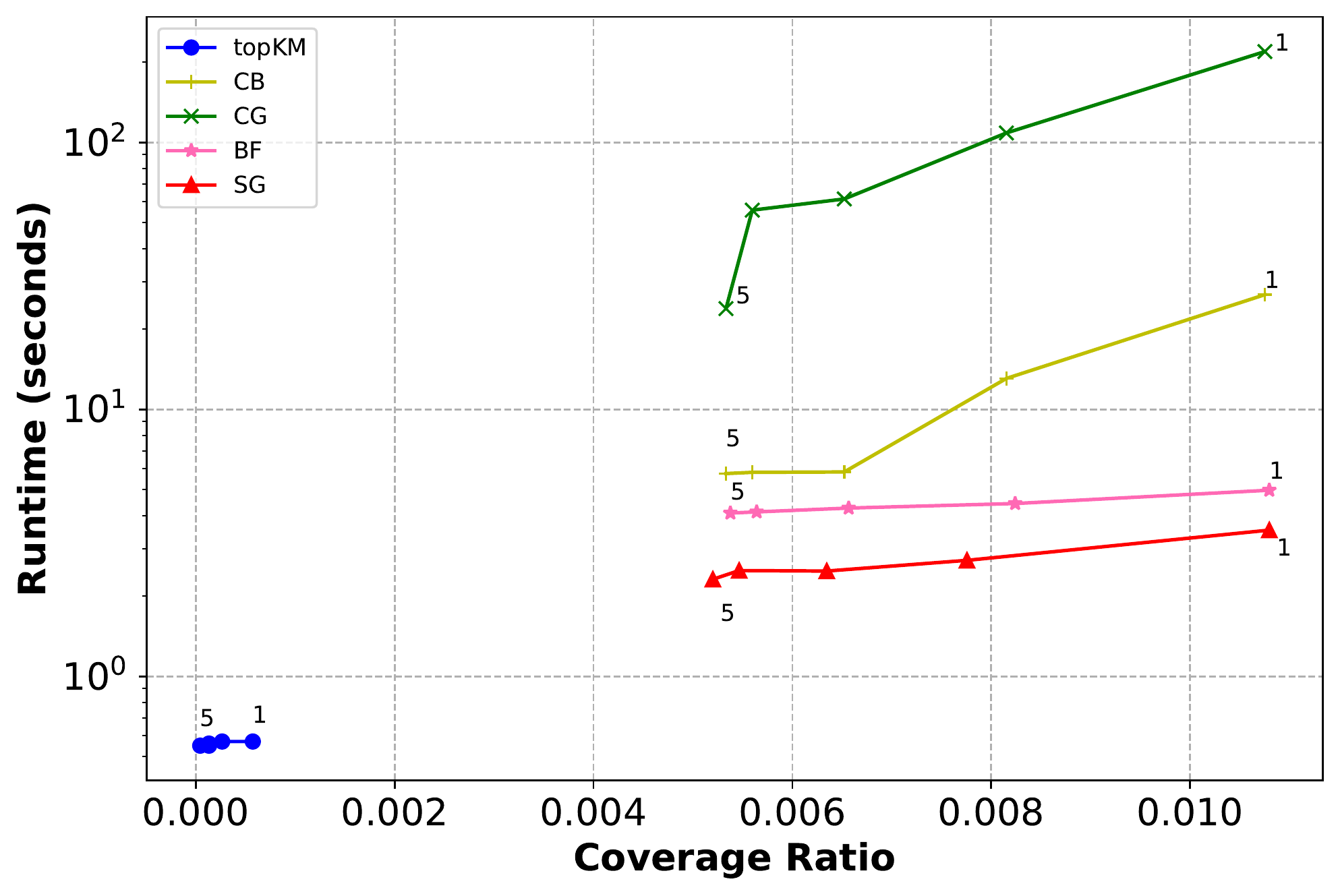}
	}
	\subfigure[Varying \spreadTime]{
		\label{fig:po:spread}
		\includegraphics[width=\graphscaleThree\textwidth]{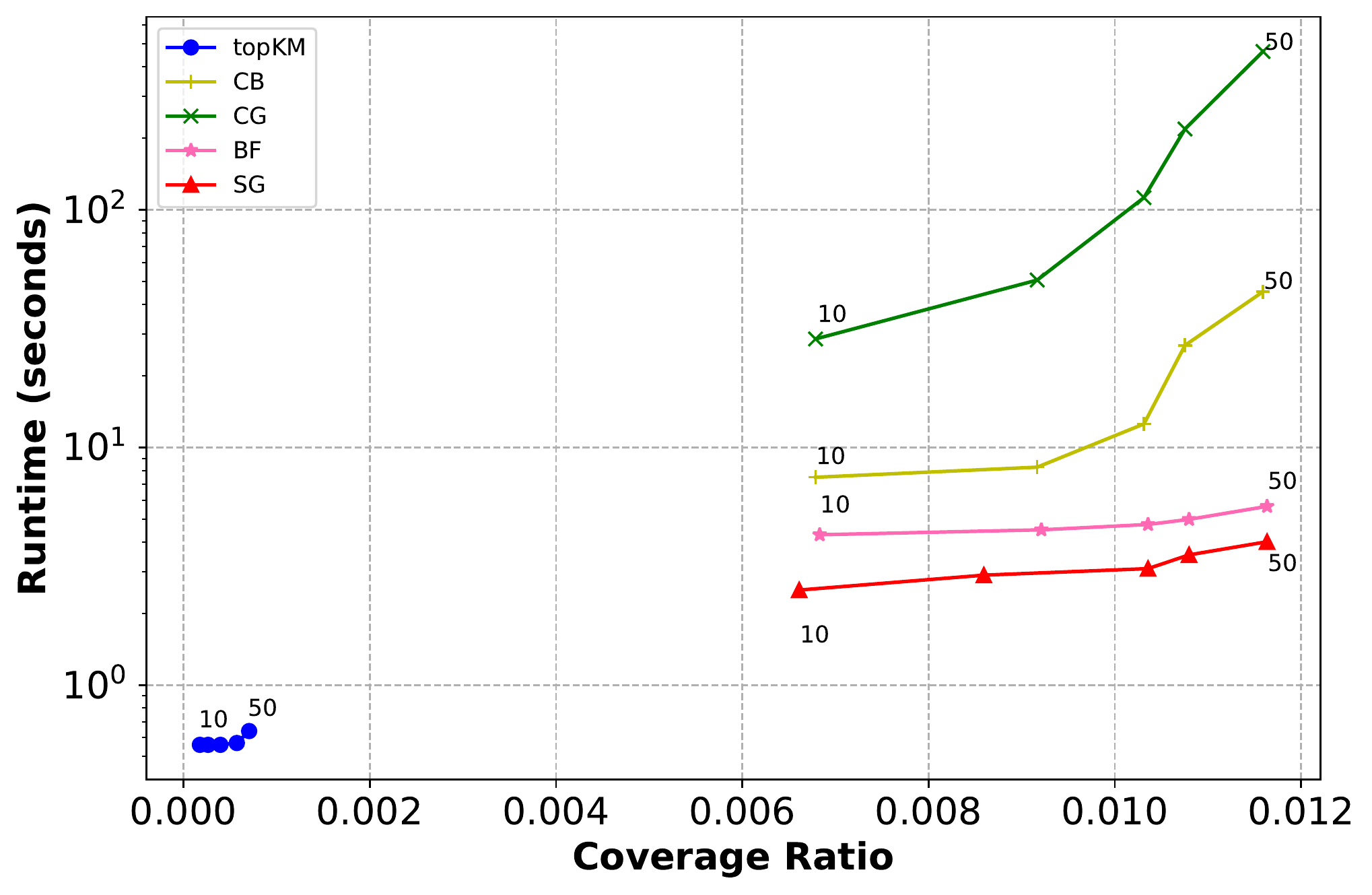}
	}
	\caption{Performance Comparison of Methods in \porto dataset.}
	\label{fig:porto}
\end{figure*}

%% file: tbl-exp-sampling.tex
\begin{table}[h]
	\small
	\renewcommand{\arraystretch}{1.2}
	\caption{Performance comparison of \bfa and \sg algorithms under different sampling sizes.}
	\label{tbl-sampling}
	\vspace{-2ex}
	\centering
	\begin{tabular}{ccccc}
		Datasets                 & Methods	& Sampling Size ($\times |\edgeSet|$)               & Runtime (seconds) & Coverage Ratio \\ \hline
		\multirow{4}{*}{\xian} 
		& \bfa				&  -               		  &0.36 	   &0.153                      \\
		& \sg				&  20\%				&0.20				  &0.140                      \\
		& \sg				&  \textbf{30\% }            &0.21 		 &0.148                      \\			
		& \sg				&  40\%             &0.42 			   	  &0.152                     \\			
		\hline
		\multirow{4}{*}{\chengdu} 
		& \bfa				&  -               		  &0.56  	   & 0.176                     \\
		& \sg				&  20\%				& 0.31		 &0.159                      \\
		& \sg				&  \textbf{30\% }            &0.43		 &0.174                      \\			
		& \sg				&  40\%             &0.63 	     &0.175                      \\			
		\hline\multirow{4}{*}{porto} 
		& \bfa				&  -               		  &4.98 	   &0.010                      \\
		& \sg				&  30\%				&2.96 		 &0.009                      \\
		& \sg				&  \textbf{40\%}             &3.73 		 &0.010                      \\			
		& \sg				&  50\%             &4.47 	     &0.010                      \\			
		\hline
	\end{tabular}
\end{table}

%% file: fig-exp-vis.tex
\begin{figure*}[h]
	\centering
	\subfigure[Congested Roads (in Red)]{
		\label{fig:vis:xa:cong}
		\includegraphics[width=\graphscaleThree\textwidth]{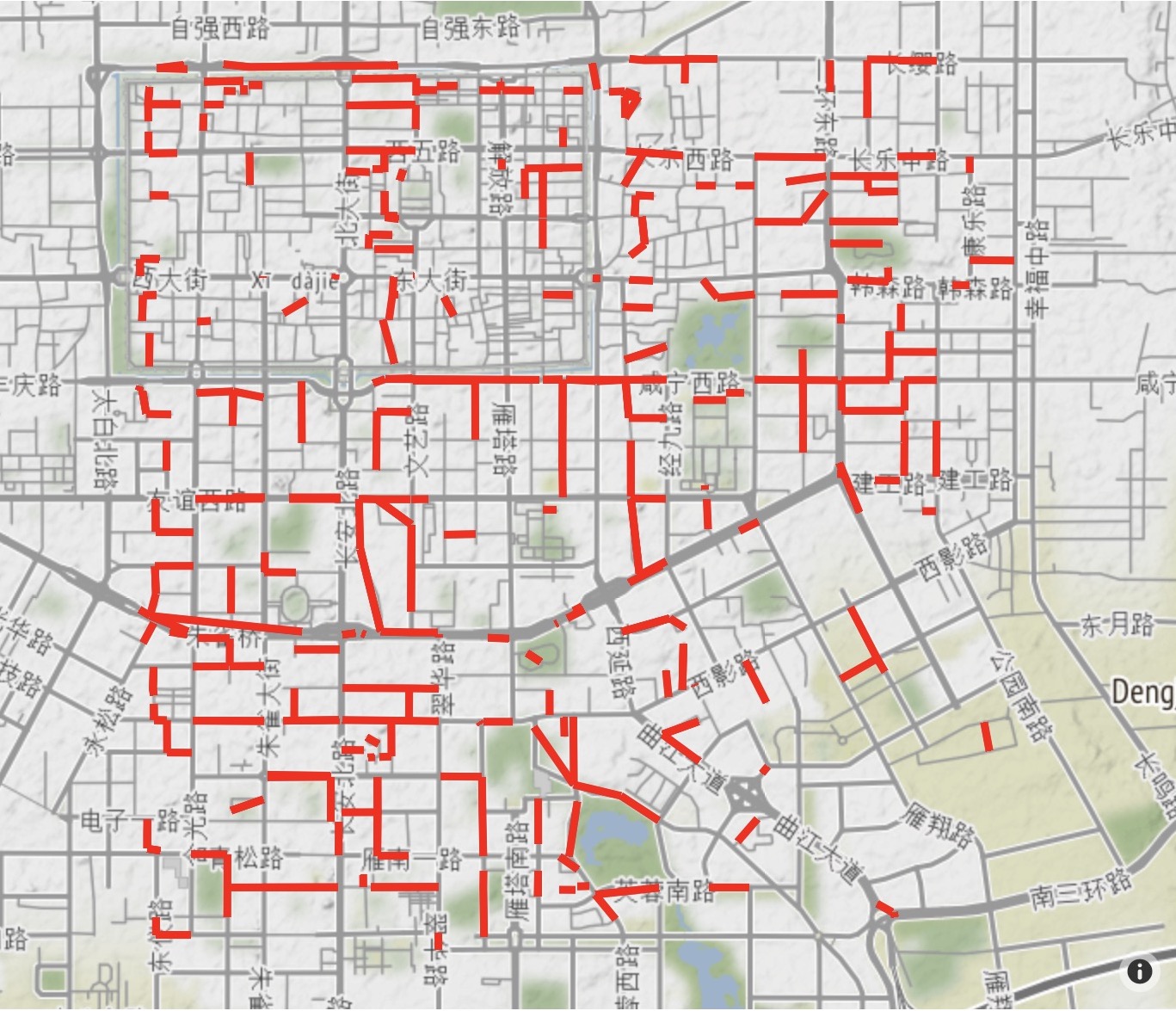}
	}
	\subfigure[\bfa]{
		\label{fig:vis:xa:bf}
		\includegraphics[width=\graphscaleThree\textwidth]{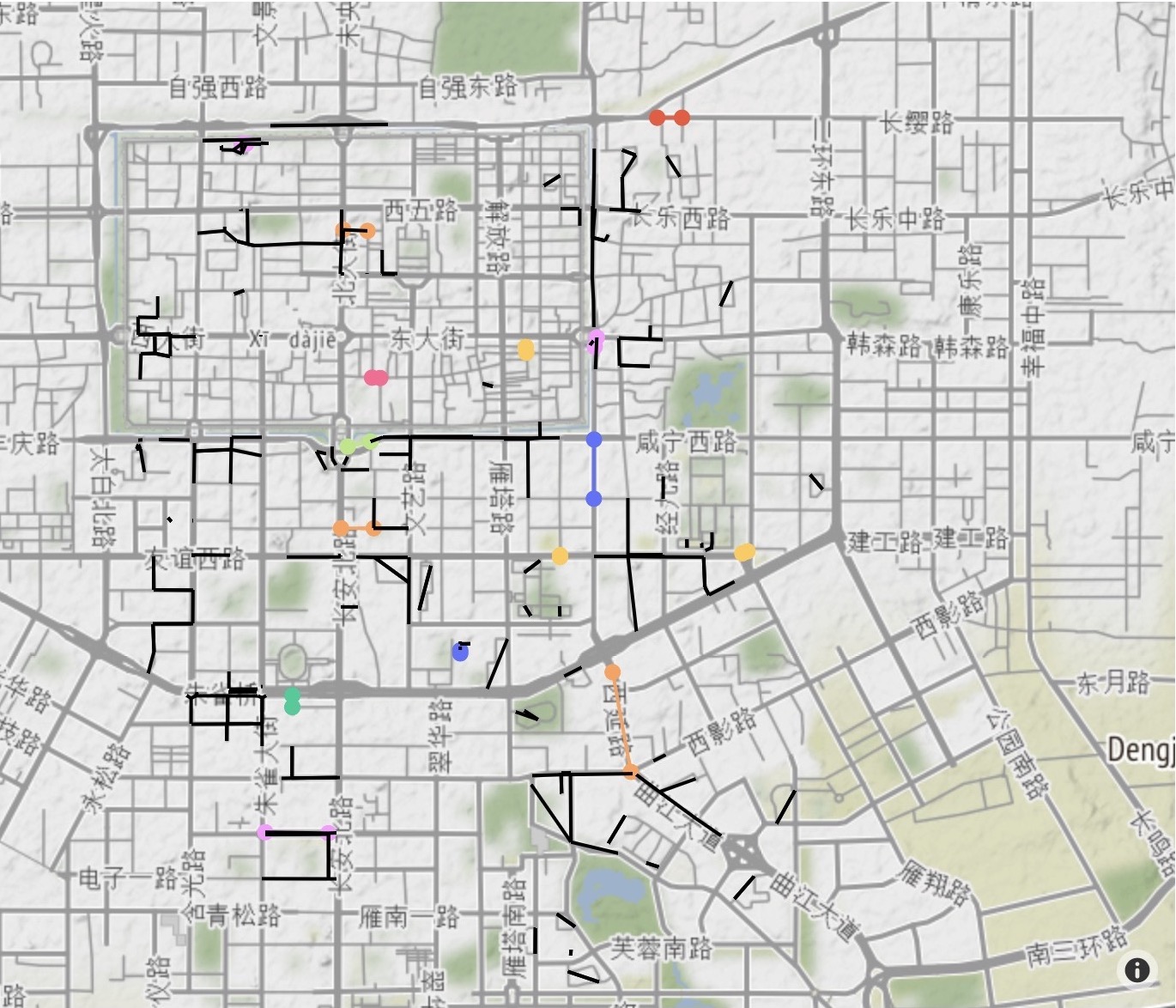}
	}
	\subfigure[\topkmin]{
		\label{fig:vis:xa:topkm}
		\includegraphics[width=\graphscaleThree\textwidth]{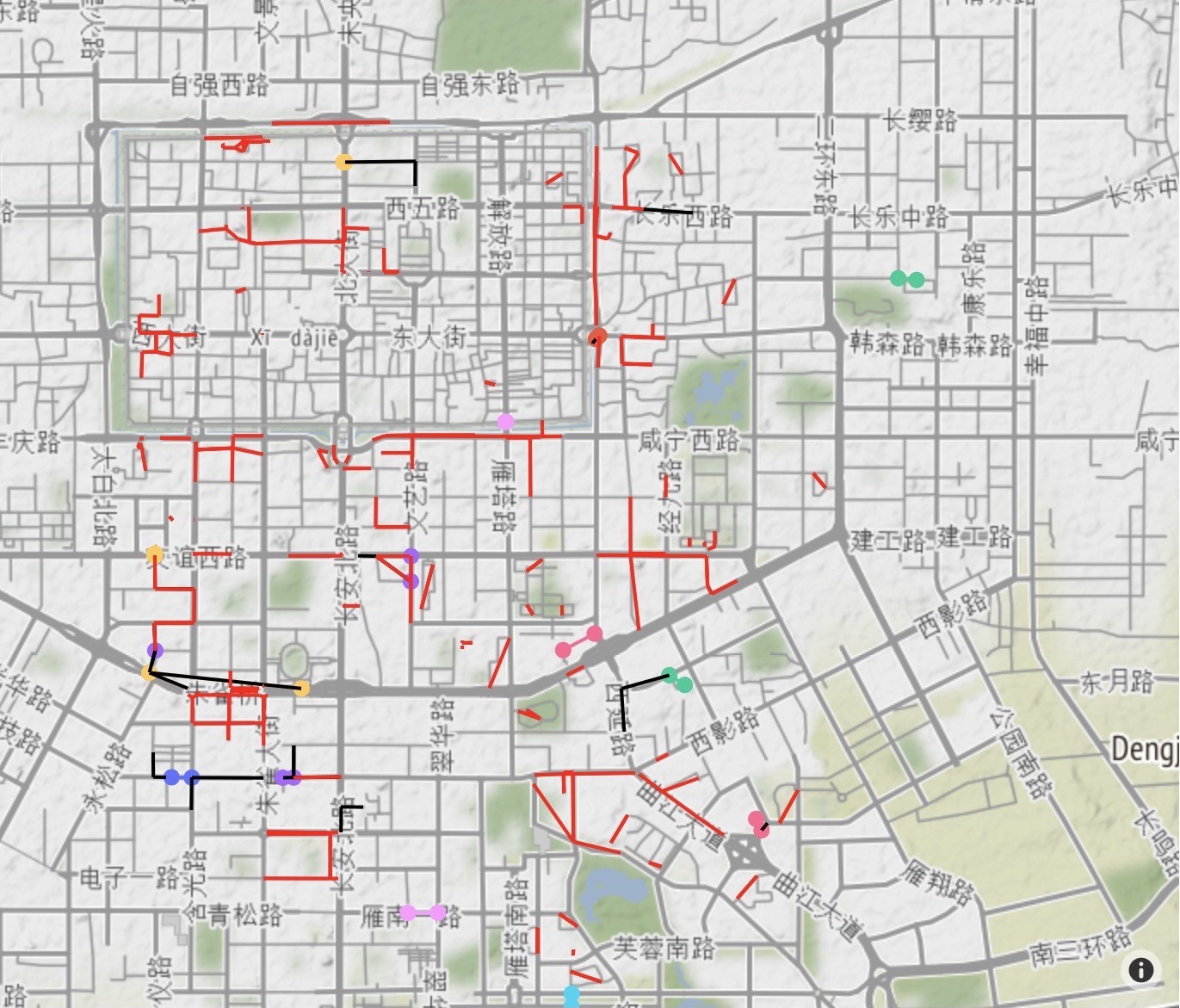}
	}
	\subfigure[\cb]{
		\label{fig:vis:xa:cb}
		\includegraphics[width=\graphscaleThree\textwidth]{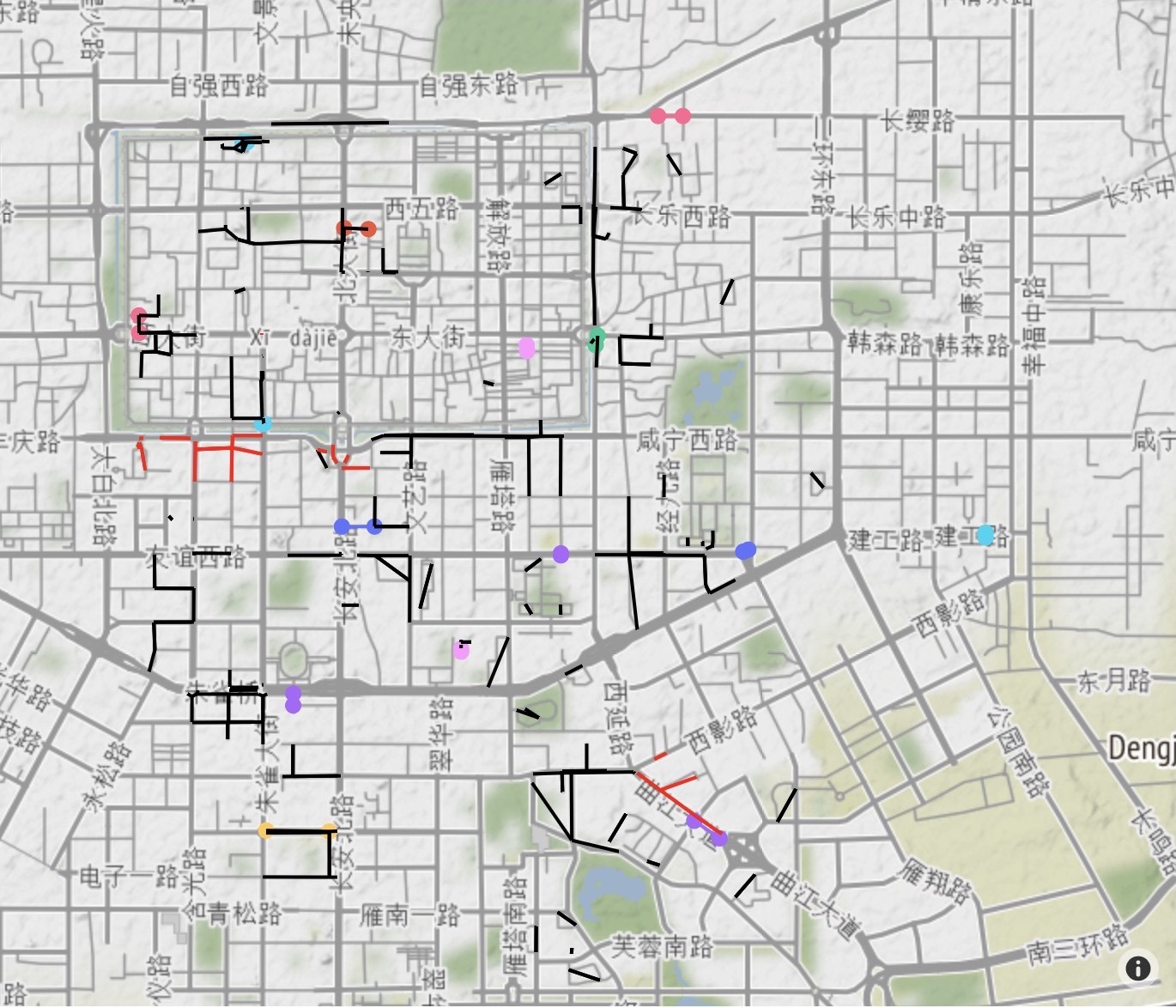}
	}
	\subfigure[\cg]{
		\label{fig:vis:xa:cg}
		\includegraphics[width=\graphscaleThree\textwidth]{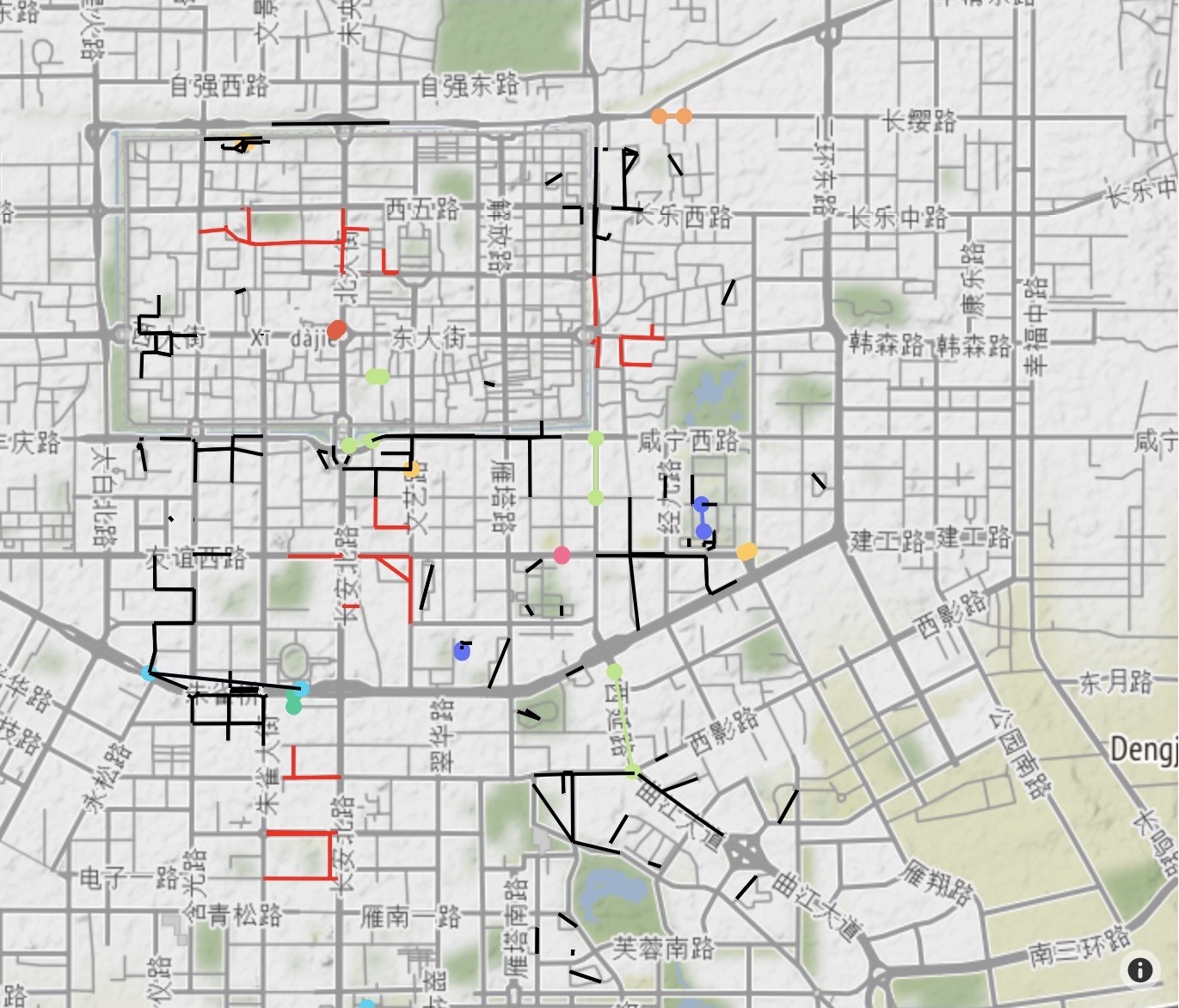}
	}
	\subfigure[\sg]{
		\label{fig:vis:xa:sg}
		\includegraphics[width=\graphscaleThree\textwidth]{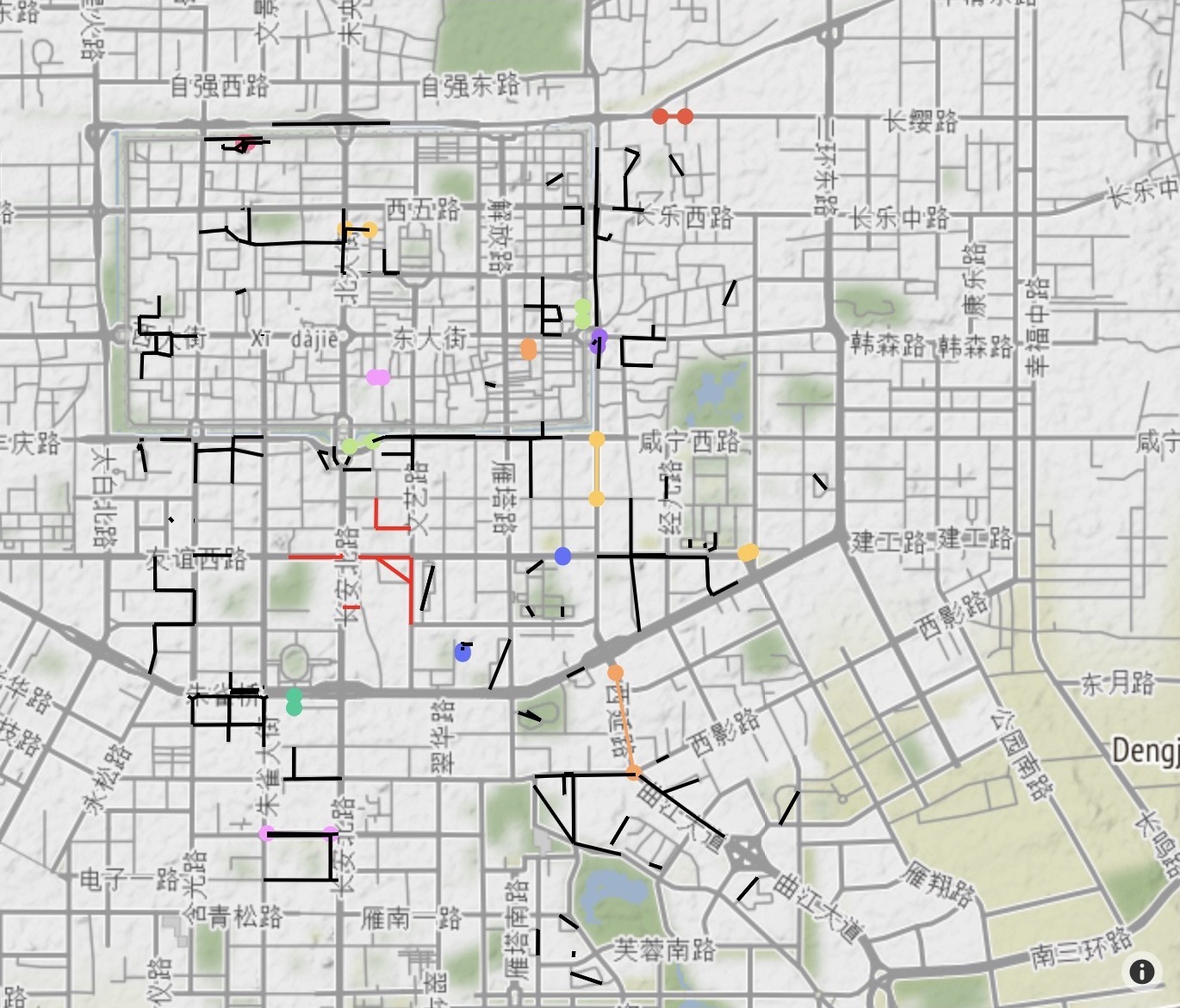}
	}
	\vspace{-2ex}
	\caption{Visualization result of  all methods in \xian dataset.}
	\label{fig:vis:xa}
\end{figure*}
\begin{figure*}[h]
	\centering
	\subfigure[Congested Roads (in Red)]{
		\label{fig:vis:cd:cong}
		\includegraphics[width=\graphscaleThree\textwidth]{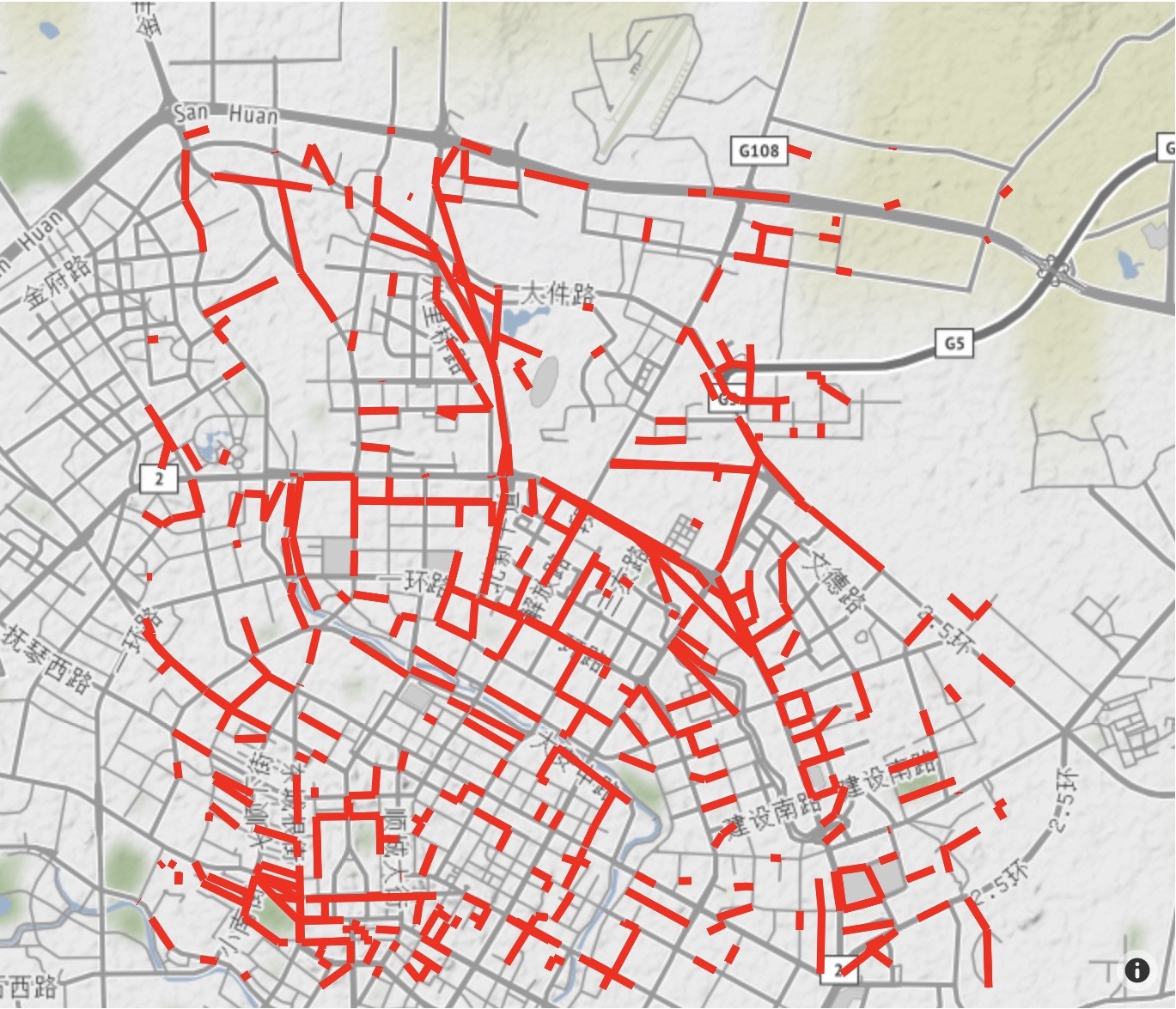}
	}
	\subfigure[\bfa]{
		\label{fig:vis:cd:bf}
		\includegraphics[width=\graphscaleThree\textwidth]{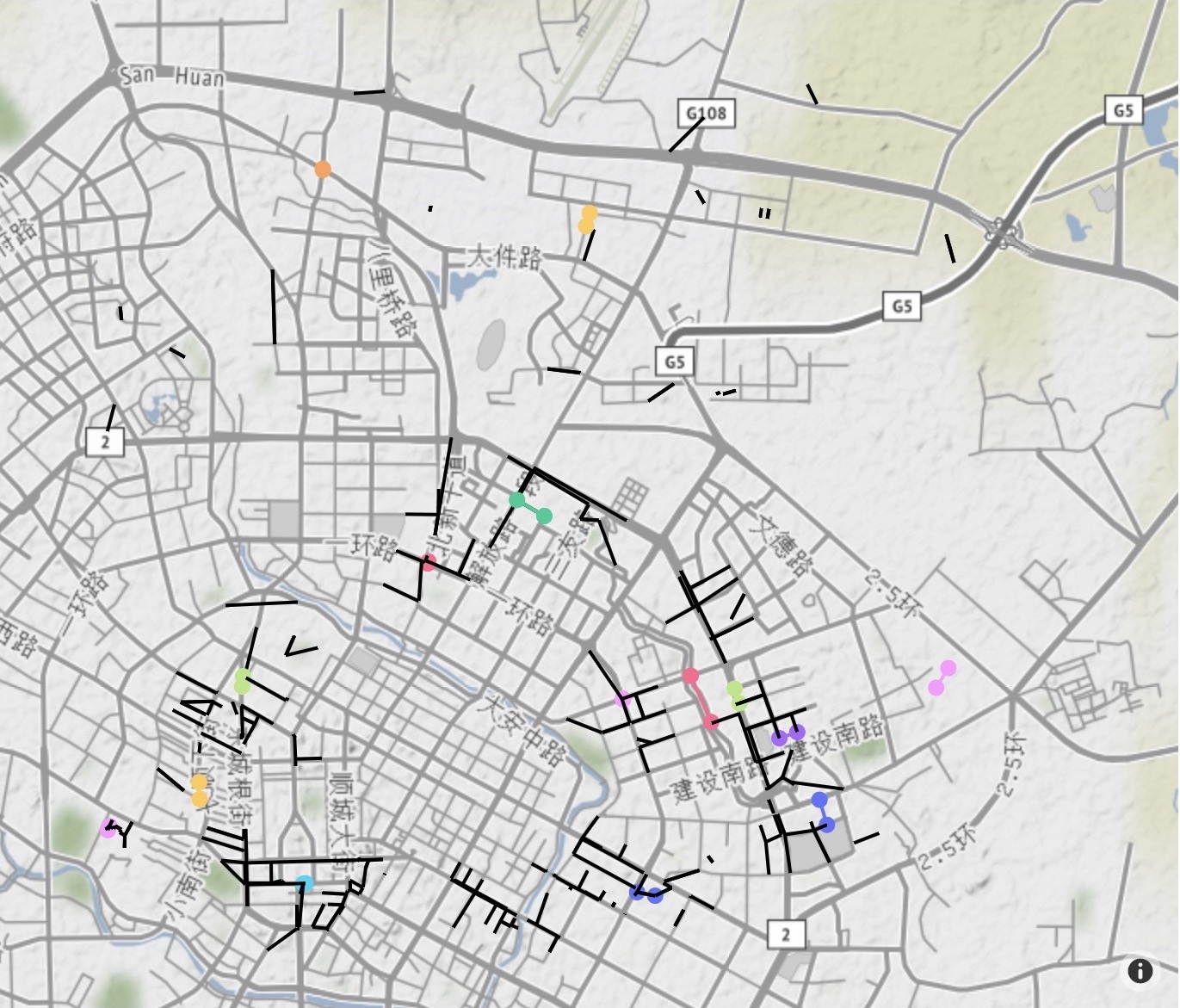}
	}
	\subfigure[\topkmin]{
		\label{fig:vis:cd:topkm}
		\includegraphics[width=\graphscaleThree\textwidth]{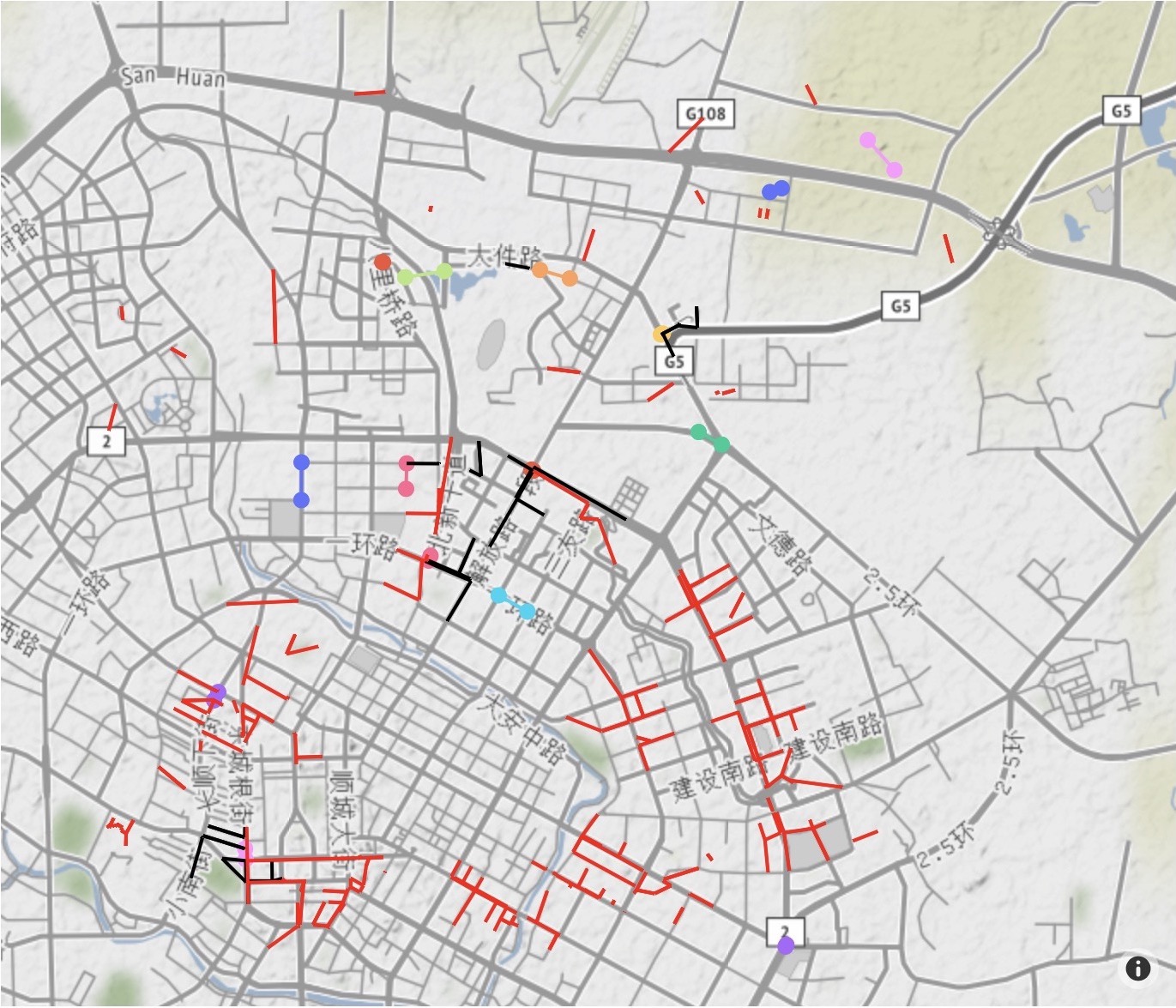}
	}
	\subfigure[\cb]{
		\label{fig:vis:cd:cb}
		\includegraphics[width=\graphscaleThree\textwidth]{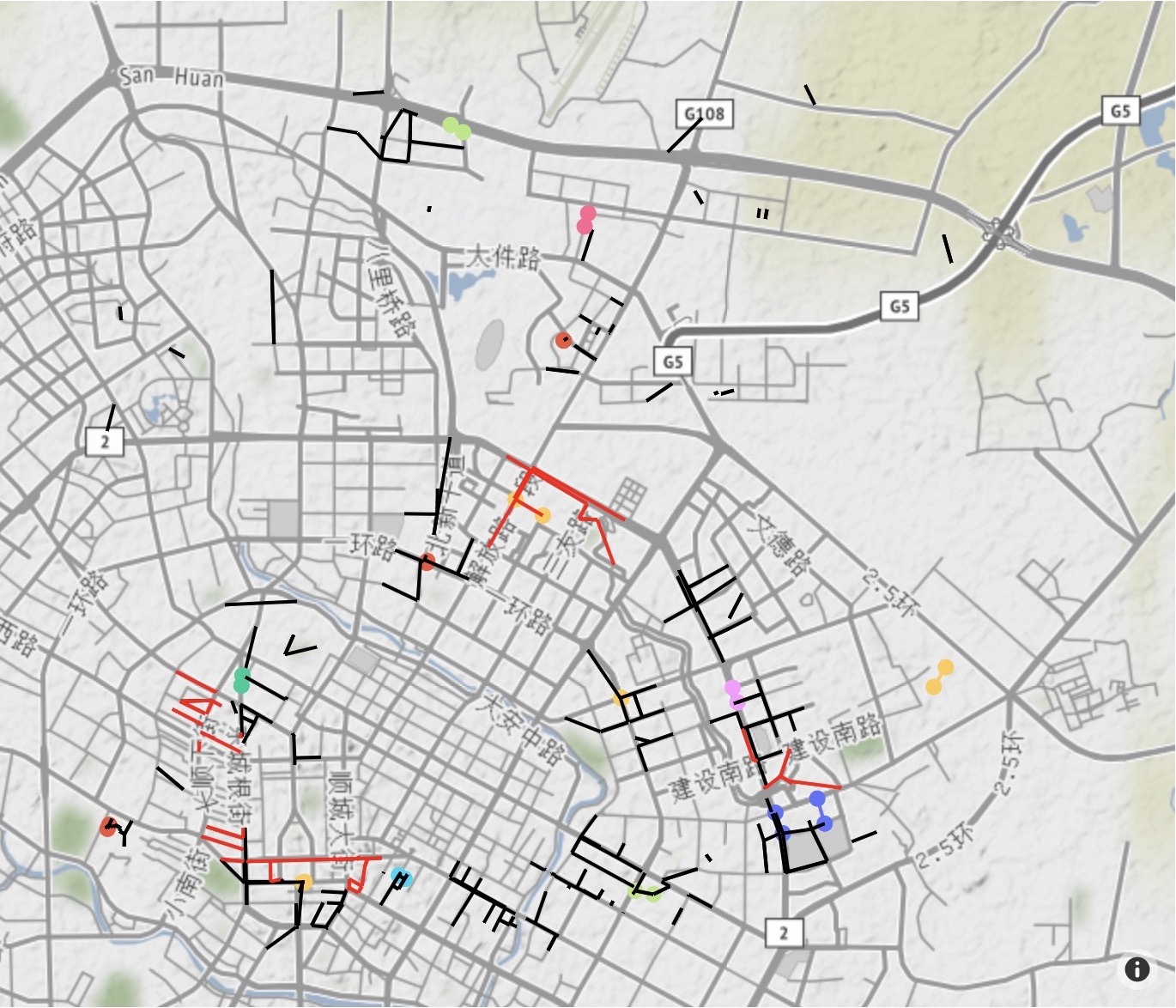}
	}
	\subfigure[\cg]{
		\label{fig:vis:cd:cg}
		\includegraphics[width=\graphscaleThree\textwidth]{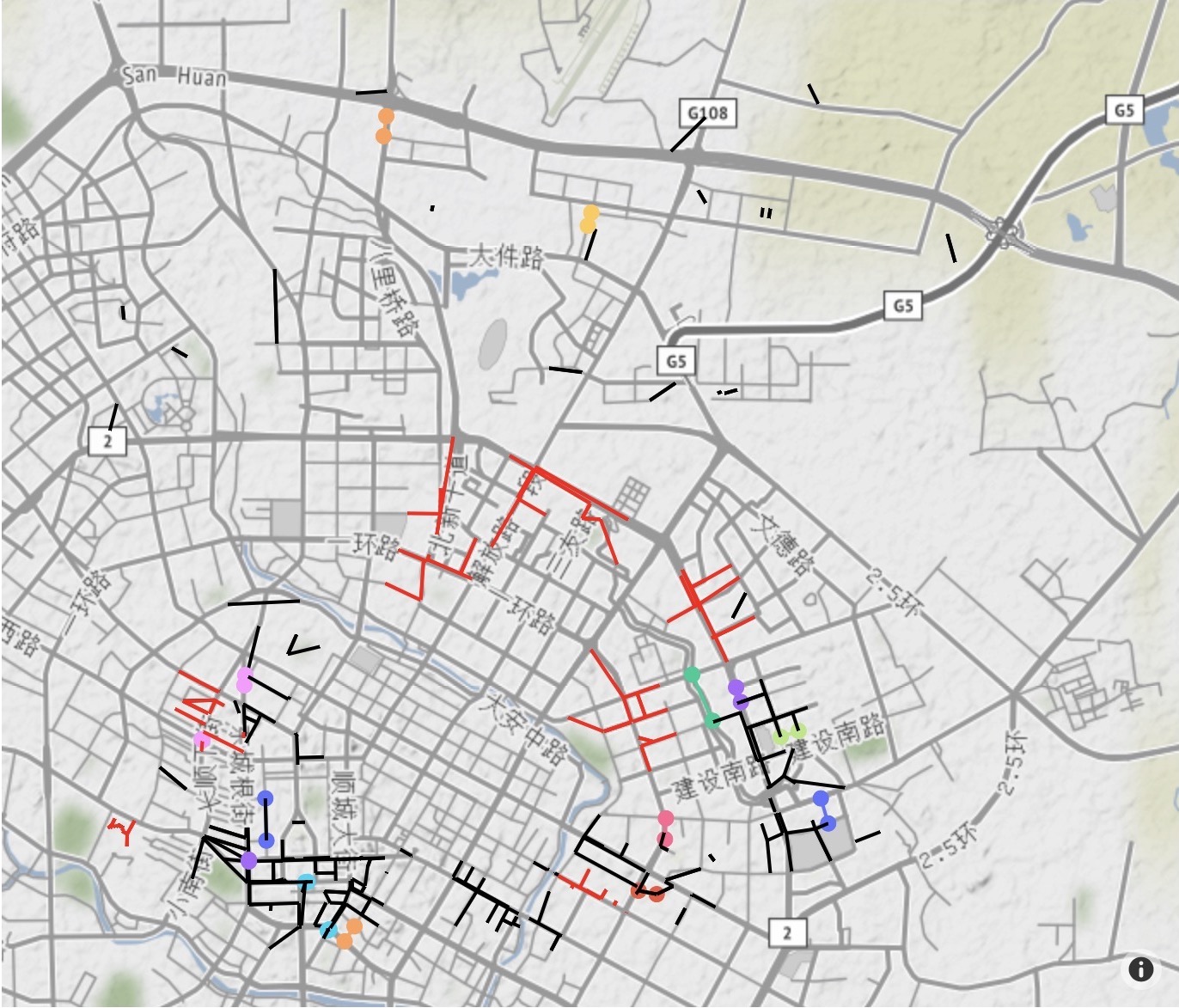}
	}
	\subfigure[\sg]{
		\label{fig:vis:cd:sg}
		\includegraphics[width=\graphscaleThree\textwidth]{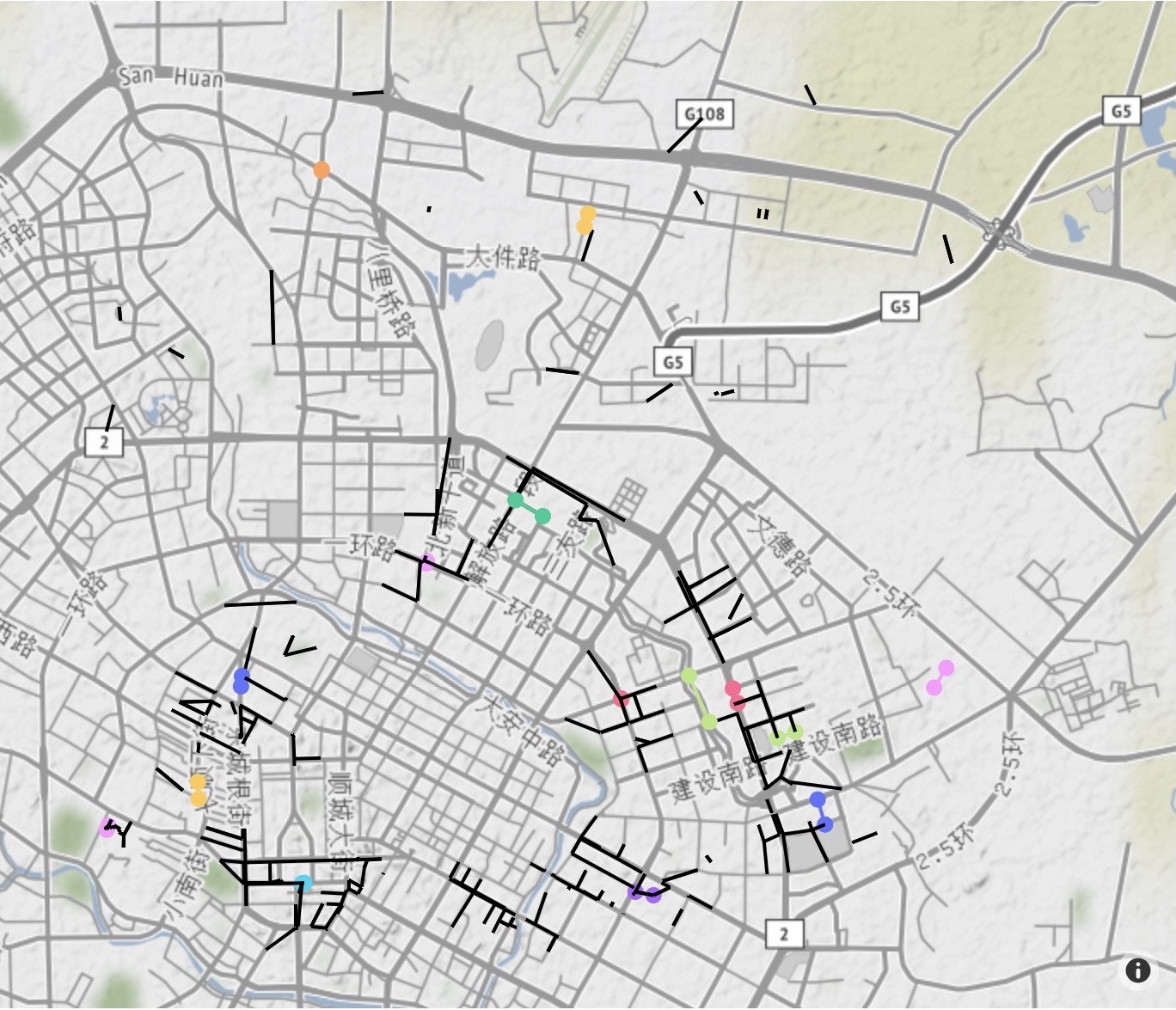}
	}
	\vspace{-2ex}
	\caption{Visualization result of  all methods in \chengdu dataset.}
	\label{fig:vis:cd}
\end{figure*}

%% file: sec6-conclusion.tex
\section{Conclusion}\label{sec-conclusion} In this paper, we have
investigated the traffic bottleneck identification problem using
trajectory datasets.
We first proposed a traffic spread model to describe traffic dynamics
over time, and used a historical trajectory dataset to provide
diffusion information over edges in the network.
Using this traffic spread model, we proposed a framework consisting
of two main phases: \phaseOne and \phaseTwo.
We then conducted an experimental study and a case study over three
real-world datasets to validate the efficiency, scalability, and
effectiveness of the proposed methods.
In future work, we would like to create a real time traffic analysis
system prototype based on our proposed methods, and use information
collected realtime from vehicles in an urban environment as
trajectories to support transportation management.
We would also like to resolve the problem of uncertainty
in real test collections, such as those created by the use of faulty
sensor data or incomplete datasets.

%% file: sec7-ack.tex
\vspace{3mm}
\myparagraph{Acknowledgement}
This research is supported in part by ARC DP200102611 and
DP190101113, Singtel Cognitive and Artificial Intelligence Lab for
Enterprises (SCALE@NTU), which is a collaboration between Singapore
Telecommunications Limited (Singtel) and Nanyang Technological
University (NTU) that is funded by the Singapore Government through
the Industry Alignment Fund - Industry Collaboration Projects Grant,
and a Tier-1 project RG114/19.